\documentclass[preprint,11pt]{elsarticle}

\usepackage{amssymb, amsmath}

 \usepackage{amsthm}

\usepackage{xspace,paralist}

\newcommand{\ceil}[1]{\left\lceil #1 \right\rceil }
\providecommand{\abs}[1]{\lvert#1 \rvert}
\newcommand{\dorg}{BRF\xspace}
\newcommand{\dorgs}{BRFs\xspace}
\newcommand{\NP}{\mathbf{NP}}
\DeclareMathOperator{\matching}{matching}

\usepackage{todonotes}
\usepackage{algorithm,algorithmicx,algpseudocode}
\newcommand{\calR}{\mathcal{R}}
\newcommand{\calRD}{\mathcal{R}_{\downarrow}}
\newcommand{\calC}{\mathcal{C}}

\newcommand{\calZ}{\mathcal{Z}}
\newcommand{\calK}{\mathcal{K}}
\newcommand{\calI}{\mathcal{I}}
\newcommand{\calS}{\mathcal{S}}
\newcommand{\vertA}{\mathbf{A}}
\newcommand{\vertB}{\mathbf{B}}
\newcommand{\vertC}{\mathbf{C}}
\newcommand{\vertD}{\mathbf{D}}
\newcommand{\ZZ}{\mathbb{Z}}
\newcommand{\RR}{\mathbb{R}}
\newcommand{\mis}{\mathbf{mis}}
\newcommand{\mhs}{\mathbf{mhs}}
\newcommand{\MIS}{\textrm{MIS}\xspace}
\newcommand{\MHS}{\textrm{MHS}\xspace}
\newcommand{\WMIS}{\textrm{WMIS}\xspace}
\newcommand{\area}{\mathrm{area}}
\newcommand{\LP}{\mathrm{LP}}
\newcommand{\QSTAB}{\textrm{QSTAB}}
\newcommand{\wit}{\mathrm{wit}}
\newcommand{\tint}{\mathrm{int}}

\newcommand{\defi}[1]{\emph{#1}}

\journal{arXiv}

\newtheorem{theorem}{Theorem}[section]

\newtheorem{lemma}[theorem]{Lemma}
\newtheorem{conjecture}[theorem]{Conjecture}
\newtheorem{proposition}[theorem]{Proposition}
\begin{document}

\begin{frontmatter}



\title{Independent sets and hitting sets of bicolored rectangular families\footnote{A preliminary version of this work appeared in IPCO 2011 under the name ``Jump Number of Two-Directional Orthogonal Ray Graphs'' \cite{soto2011jump}.}}


\author[l1]{Jos\'e A.~Soto}
\ead{jsoto@dim.uchile.cl}
\address[l1]{Department of Mathematical Engineering and Center for Mathematical Modeling, Universidad de Chile}

\author[l2]{Claudio Telha}
\ead{claudio.telha@uclouvain.be}
\address[l2]{Universit\'e catholique de Louvain. }
\begin{abstract}
A bicolored rectangular family \dorg is a collection of all axis-parallel rectangles contained in a given region $\calZ$ of the plane formed by selecting a bottom-left corner from a set $A$ and an upper-right corner from a set $B$. We prove that the maximum independent set and the minimum hitting set of a \dorg have the same cardinality and devise polynomial time algorithms to compute both. As a direct consequence, we obtain the first polynomial time algorithm
to compute minimum biclique covers, maximum cross-free matchings and jump numbers in a class of bipartite graphs that significantly extends convex bipartite graphs and interval bigraphs. We also establish several connections between our work and other seemingly unrelated problems. Furthermore, when the bicolored rectangular family is weighted, we show that the problem of finding the maximum weight of an independent set is $\NP$-hard, and provide
efficient algorithms to solve it on certain subclasses. 
\end{abstract}


\end{frontmatter}



\section{Introduction}

Suppose we are given a collection of axis-parallel closed rectangles in the plane. A subcollection of rectangles that do not pairwise intersect is
called an \emph{independent set}, and a collection of points in the plane intersecting (hitting) every rectangle is called a \emph{hitting set}. In
this paper we study the problems of finding a maximum independent set (\MIS) and a minimum hitting set (\MHS) of restricted classes of rectangles
arising from bicolored point-sets in the plane, and relate them to other problems in graph and poset theory.

Both the \MIS and its weighted version \WMIS are important problems in computational geometry with
a variety of applications \cite{DoerschlerF1992,LentSW1997,AgarwalKS1998,FukudaMMT2001,LewinNO2002}.
Since \MIS is $\NP$-hard~\cite{FowlerPT1981,ImaiA1983}, a significant amount of research has been devoted to heuristics and
approximation algorithms. Charlermsook and Chuzhoy~\cite{ChalermsookC2009}, and Chalermsook~\cite{Chalermsook2011} describe two different 
$O(\log \log m)$-approximation algorithms for the \MIS problem on a family of $m$ rectangles, while Chan and Har-Peled~\cite{ChanHP2012} provide an
$O(\log m / \log\log m)$-approximation factor for \WMIS. The approximation factor achieved by these polynomial time
algorithms are the best so far for general rectangle families. Nevertheless, very recently Adamaszek and Wiese~\cite{AdamaszekW2013} presented a
pseudo-polynomial algorithm achieving a $(1+\varepsilon)$-approximation for \WMIS in general rectangles. For special classes of rectangle families, the
situation is better: there are polynomial time approximation schemes (PTAS) for \MIS in squares~\cite{Chan2003}, rectangles with bounded
width to height ratio~\cite{ErlebachJS2001}, rectangles of constant height~\cite{AgarwalKS1998}, and rectangles forming a pseudo-disc family, that is,
the intersection of the boundaries of any two rectangles consists of at most two points~\cite{AgarwalM2006, ChanHP2012}.

The \MHS problem is the dual of the \MIS problem, and therefore the value of an optimal solution of \MHS is an upper
bound for that of \MIS. The \MHS problem is also $\NP$-hard~\cite{FowlerPT1981}, but there are also PTAS for special cases, including
squares~\cite{Chan2003}, rectangles of constant height~\cite{ChanM2005} and pseudodiscs~\cite{MustafaR2010}.
More recently, Aronov, Ezra, and Sharir~\cite{AronovES2010} proved the existence of $\varepsilon$-nets of size $O(\frac{1}{\varepsilon}\log\log\frac{1}{\varepsilon})$ for axis-parallel rectangle families.
Using Br{\"o}nnimann and Goodrich's  technique~\cite{BronnimannG1995}, this yields an
$O(\log \log \tau)$ approximation algorithm for any rectangle family that can be hit by at most~$\tau$ points yielding the best
approximation guarantee for \MHS known.

\subsection{Main results}
In this article we study both the \MIS and \MHS problems on a special class of rectangle families. Given two finite sets $A,B \subseteq \RR^2$ and an
arbitrary set $\calZ \subseteq \RR^2$ of the plane, we define the \defi{bicolored rectangle family} (\dorg) induced by $A$, $B$ and $\calZ$ as the set 
$\calR=\calR(A,B,\calZ)$ of all rectangles contained in $\calZ$, having bottom-left corner in $A$ and top-right corner in $B$. When $\calZ$ is the
entire plane, we write $\calR(A,B)$ and we call it an unrestricted \dorg.

Our main result is an algorithm that simultaneously finds an independent set $I$ and a hitting set $H$ of $\calR(A,B,\calZ)$ with
$\abs{I}=\abs{H}$. By linear programming duality, $I$ and $H$ are a \MIS and a \MHS of $\calR$, respectively; meaning in particular
that we have a max-min relation between both quantities.

Our algorithm is efficient: if $n$ is the number of points in $A\cup B$, both the \MIS and the \MHS of $\calR$ can
be computed deterministically in $O(n^{2.5}\sqrt{\log n})$ time. If we allow randomization, we can compute them with high probability in time $O((n\log n)^{\omega})$, where $\omega < 2.3727$ is the exponent of square matrix multiplication. Here we implicitly assume that testing if a
rectangle is contained in $\calZ$ can be done in unit time (by an oracle); otherwise, we need additional $O(n^2T(\calZ))$ time, where $T(\calZ)$ is
the time for testing membership in~$\calZ$. The bottleneck of our algorithm is the computation of a maximum matching on a bipartite graph needed for
an algorithmic version of the classic Dilworth's theorem.

We also show that a natural linear programming relaxation for the \MIS of \dorgs can be used to compute the optimal size of the maximum
independent set. Using an uncrossing technique, we prove that one of the vertices in the optimal face of the underlying polytope is integral, and we
show how to find that vertex efficiently. This structural result about this linear program relaxation gives a second algorithm to compute the maximum
independent set of a \dorg and should be of interest by itself.

\subsection{Biclique covers, cross-free matchings and jump number}
Our results have some consequences for two graph problems called the minimum biclique cover and the maximum cross-free matching. Before stating the
relation we give some background on these problems.

A \defi{biclique} of a bipartite graph is the edge set of a complete bipartite subgraph. A \defi{biclique cover} is a collection of bicliques whose union is the entire edge set. Two edges $e$ and $f$ \defi{cross} if there is a biclique containing both. A \defi{cross-free matching} is a collection of pairwise non-crossing edges.

The problem of finding a minimum biclique cover arises in many areas (e.g.~biology~\cite{NauMWA1978}, chemistry~\cite{CohenS99} and communication complexity~\cite{KushilevitzN97}). Orlin~\cite{Orlin77} has shown that finding a minimum biclique cover of a bipartite graph is $\NP$-hard. M\"uller~\cite{Muller96} extended this result to chordal bipartite graphs. To our knowledge, the only classes of graphs for which this problem has been explicitly shown to be polynomially solvable before our work are $C_4$-free bipartite graphs~\cite{Muller96}, distance hereditary bipartite graphs~\cite{Muller96}, bipartite permutation graphs~\cite{Muller96}, domino-free bipartite graphs~\cite{AmilhastreJV97} and convex
bipartite graphs~\cite{Gyori84,FranzblauK84}. The proof of polynomiality for this last class can be deduced from the statement that the minimum biclique cover problem on convex bipartite graphs is equivalent to finding the minimum basis of a family of intervals. The latter was originally studied and solved by Gy\H{o}ri~\cite{Gyori84}, but its connection with the minimum biclique problem was only noted several years later~\cite{AmilhastreJV97} (See Section~\ref{sec:discussion} for more details).

The maximum cross-free matching problem is often studied because of its relation to the jump number problem in Poset Theory. The \defi{jump number} of
a partial order $P$ with respect to a linear extension $L$ is the number of pairs of consecutive elements in $L$ that are incomparable in $P$. The
jump number of $P$ is the minimum of this quantity over all linear extensions. Chaty and Chein~\cite{ChatyC1979} show that computing the jump
number of a poset is equivalent to finding a maximum alternating-cycle-free matching in the underlying comparability graph, which is $\NP$-hard as
shown by Pulleyblank~\cite{Pulleyblank82}. For chordal bipartite graphs, alternating-cycle-free matching and cross-free matchings coincide, making the
jump number problem equivalent to the maximum cross-free matching problem. M\"uller~\cite{Muller90} has shown that this problem is $\NP$-hard for
chordal bipartite graphs, but there are efficient algorithms to solve it on important subclasses. In order of inclusion, there are linear and 
cuadratic 
algorithms for  bipartite permutation graphs~\cite{SteinerS1987,Branstadt1989,Fauck1991}, a cuadratic algorithm for biconvex graphs~\cite{Branstadt1989} and an $O(n^9)$ time algorithm for convex bipartite graphs~\cite{Dahlhaus1994}.

To relate our results with the problems just defined consider the following construction. Given a \dorg $\calR(A,B,\calZ)$, create a bipartite graph
with vertex color classes $A$ and $B$ identifying every rectangle $R\in \calR$ with bottom-left corner $a\in A$ and top-right corner $b\in B$ with an
edge connecting $a$ and $b$. The resulting graph $G=(A\cup B, \calR)$ is the \defi{graph representation} of $\calR$. We call the graphs arising in
this way \defi{\dorg graphs}.

It is an easy exercise to check that for an \defi{unrestricted \dorg graph}---i.e., a graph $G=(A\cup B, \calR)$ where $\calR$ is unrestricted---two
edges $R$ and $R'$ cross if and only if $R$ and $R'$ intersect as rectangles.  In particular, the cross-free matchings of $G$ are in correspondence
with the independent sets of $\calR$. Similarly, the maximal bicliques of $G$ (in the sense of inclusion) are exactly the maximal families of pairwise
intersecting rectangles in $\calR$. Using the Helly property\footnote{If a collection of rectangles pairwise intersect, then there is a point hitting
all of them.} for axis-parallel rectangles we conclude that the minimum hitting set problem on $\calR$ is equivalent to the minimum biclique
cover of $G$.

Our results imply then that for unrestricted \dorg graphs, the maximum size of a cross-free matching equals the minimum size of a biclique cover and both optimizers can be computed in polynomial time. Since unrestricted \dorgs are chordal bipartite~\cite{ShresthaTU2010}, we also obtain polynomial time algorithms to compute the jump number of unrestricted \dorgs.

\subsection{Additional results for the weighted case}
~
We also consider the natural weighted version of \MIS, denoted \WMIS, where each rectangle in the family has a non-negative weight, and we aim to find a collection of disjoint rectangles with maximum total weight. For \dorgs, this problem is equivalent to finding a maximum weight cross-free matching of the associated \dorg graph.

We show that this problem is $\NP$-hard for unrestricted \dorgs and weights in $\{0,1\}$. Afterwards,  we present some results for certain subclasses. For bipartite permutation graphs, we provide an $O(n^2)$ algorithm for arbitrary weights and a specialized $O(n)$ algorithm when weights are in $\{0,1\}$.
  We also note how the algorithm of Lubiw~\cite{Lubiw1991} for the maximum weight set of point-interval pairs readily translates into an algorithm for the \WMIS of convex graphs that runs in $O(n^3)$ time. Recent results of Correa et al.~\cite{CorreaFS14, CorreaFPS14} can be used to extend the $\NP$-hardness to interval bigraphs and to provide a 2-approximation algorithm for \WMIS on this class.

\subsection{Relation with other works}

The min-max result relating independent sets and hitting sets in \dorgs can also be obtained as a consequence of a deep duality result of Frank and Jord\'an for set-pairs \cite{FrankJ95}. Because of the many non-trivial connections between our work and other problems indirectly related to set-pairs, we defer the introduction of this concept and the discussion of these connections to Section~\ref{sec:discussion}, when all our results have been introduced. Although the min-max result is a consequence of an existent result, our algorithmic proof is significantly simpler than those for the larger class of set-pairs. And because of its geometrical nature, it is also more intuitive.

\section{Preliminaries}

We denote the coordinates of the plane $\RR^2$ as $x$ and $y$, so that a point~$p$ is written as $(p_x,p_y)$. Given two points $p$, $p'
\in \RR^2$, we write $p \leq_{\RR^2} p'$ if and only if $p_x \leq p'_x$ and
$p_y \leq p'_y$. For any set $S \subseteq \RR^2$, the \defi{projection} $\{s_x\colon s\in S\}$ of $S$
onto the $x$ axis is denoted by $S_x$.  Given two sets $S$ and  $S'$ in
$\RR^2$, we write $S_x < S'_x$ if the projection $S_x$ is to the left of the
projection $S'_x$, that is, if $p_x < p'_x$ for all $p\in S$, $p'\in S'$. We
extend this convention to $S_x > S'_x, S_x \leq S'_x$ and $S_x \geq S'_x$, as
well as to the projections onto the $y$-axis.
For our purposes, a \emph{rectangle}~$R$ is the cartesian product of two closed
intervals. In other words, we only consider axis-parallel closed rectangles.
We say that two rectangles $R$ and $R'$ \defi{intersect} if they have
a non-empty geometric intersection. A point $p$ \defi{hits} a rectangle $R$ if
$p\in R$.  A collection of  pairwise non-intersecting rectangles is called an
\defi{independent set of rectangles}. A collection of points $H$ is a
\defi{hitting set} of a rectangle family~$\calC$ if each rectangle in
$\calC$ is hit by a point in $H$. We denote by $\mis(\calC)$ and $\mhs(\calC)$
the sizes of a maximum independent set of rectangles in $\calC$ and a minimum hitting
set for $\calC$ respectively.

The \defi{intersection graph} $\calI(\calC)$ of a collection of rectangles $\calC$
is the graph on $\calC$ having  edges between intersecting rectangles:
\begin{align*}
\calI(\calC)\equiv (\calC , E), \quad \text{ where } E=\{RR' \colon  R
\cap R' \neq \emptyset \}.
\end{align*}
Note that the independent sets of rectangles in $\calC$ are exactly the stable sets of $\calI(\calC)$. Furthermore, since $\calC$ has the Helly property, we can assign to every clique in $\calI(\calC)$ a unique \emph{witness point}, defined as the leftmost and lowest point contained in all rectangles of the clique. Since different maximal cliques have different witness points, it is easy to prove that $\calC$ admits a minimum hitting set consisting only of witness points of maximal cliques. In particular, $\mhs(\calC)$ equals the minimum size of a clique-cover of $\calI(\calC)$.

For both \MIS and \MHS, we can restrict ourselves to the family $\calC_\downarrow$ of inclusionwise minimal rectangles in $\calC$: any maximum independent set in $\calC_\downarrow$ is also maximum in $\calC$ and any minimum hitting set for $\calC_\downarrow$ is also minimum for $\calC$. Since the size of every independent set is at most the size of any hitting set, we observe that for any family $\calC$,
\begin{align}
\mis(\calC_\downarrow)=\mis(\calC) \leq \mhs(\calC)=\mhs(\calC_\downarrow).
\end{align}

\subsection{Bicolored Rectangular Families (\dorgs)}

In what follows, let $A$ and $B$ be finite sets of white and gray points on the plane, respectively, and $\calZ \subseteq \RR^2$ be a set not
necessarily finite.

We denote by $\Gamma(a,b)=\{p\in \RR^2\colon a_x \leq
p_x \leq b_x, a_y\leq p_y \leq b_y \}$ the rectangle with bottom-left corner $a$ and
upper-right corner $b$.  The set
$$\calR=\calR(A,B,\calZ)=\{\Gamma(a,b)\colon a\in A, b\in B, a\leq_{\RR^2}
b \text{ and } \Gamma(a,b) \subseteq \calZ\}$$
is called the bicolored rectangular family (\dorg) associated to $(A,B,\calZ)$. 

We denote $n=\abs{A\cup B}$ and use $[k]$ to denote the set $\{1,2, \ldots, k\}$. To make the exposition simpler, we will assume without loss
of generality that $A\cap B=\emptyset$, that $A \cup B \subseteq \calZ$, that the points in $A\cup B$ have integral coordinates\footnote{This can easily be done by translating the plane and applying  piecewise linear transformation on the axis} in $[n]\times [n]$, and that
no two points of $A \cup B$ share a common coordinate value. Under these assumptions, we only need to consider hitting sets in $\calZ\cap [n]^2$.  We
also assume that the set $\calZ$ is given implicitly, i.e.~its size is not part of the input, and that we can test if a rectangle is contained in
$\calZ$ either in unit time (using an oracle) or in a fixed time $T(\calZ)$.

A non-empty intersection between two rectangles $R$ and $S$ in $\calR$ is called a \defi{corner intersection} if either rectangle contains a
vertex of the other rectangle in its interior. Otherwise, the intersection is
called \defi{corner-free}. A \defi{corner-free intersection} (c.f.i.) family is a collection of rectangles
such that every intersection is corner-free. These intersections are generically
shown in
Figure \ref{jump-fig22a}.
\begin{figure}[ht]
\centering
\includegraphics[scale=1]{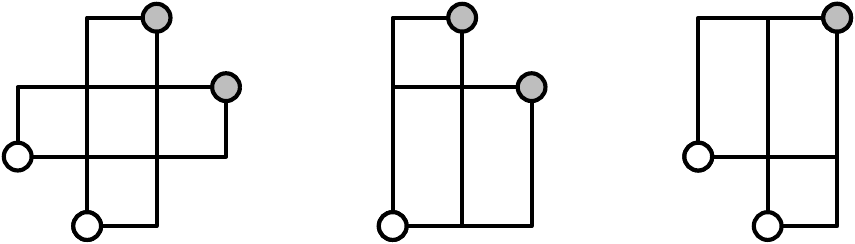}
\caption{Corner-free intersections}\label{jump-fig22a}
\end{figure}
Corner-free intersections families have a special structure, as the next
proposition shows.

\begin{proposition}\label{prop:cfi-comparability}
If $\calK$ is a c.f.i.~family then $\calI(\calK)$ is a comparability graph.
\end{proposition}
\begin{proof}

Consider the  relation $\hookrightarrow$ on $\calK$.
given by $R \hookrightarrow S$ if and only if $ R_x \subseteq S_x$ and $S_y
\subseteq R_y$. It is easy to check that $\hookrightarrow$ is a partial order
relation. Since $\calK$ is a c.f.i.~family,  $R$ and $S$ intersect if and only
if they are comparable under~$\hookrightarrow$. Therefore $\calI(\calK)$  is the
comparability graph of $(\calK, \hookrightarrow$).
\end{proof}

The following lemma will also be useful later.

\begin{lemma}\label{lema:inclusionwiseminimal}
Every rectangle in $\calR_{\downarrow}$ does not contain any point in $A\cup B$
other than its two defining corners.
\end{lemma}
\begin{proof}
Direct.
\end{proof}

\subsection{Comparability intersection graphs }

Let $\calK$ be a family of rectangles such that $\calI(\calK)$ is the
comparability graph of an arbitrary partial order $(\calK,\hookrightarrow)$. Independent
sets in $\calK$ correspond then to antichains in $(\calK,\hookrightarrow)$;
therefore the maximum independent set problem in $\calK$ is equivalent to
the \defi{maximum cardinality antichain problem} in $(\calK,\hookrightarrow)$.

Rectangles hit by a fixed point are trivially pairwise intersecting; therefore they are chains in
$(\calK,\hookrightarrow)$. By the Helly property, any family of pairwise intersecting rectangles is hit by single point.
It follows that finding a minimum hitting set of $\calK$ is equivalent to finding a minimum size family of chains in $(\calK,\hookrightarrow)$ covering all the elements of $\calK$. The latter is the description of the \defi{minimum chain covering problem}.

Dilworth's theorem states that for any partial order, the maximum
cardinality of an antichain equals the size of a minimum chain cover. In our
context, this directly translate to:
\begin{lemma}\label{lema:minmaxbasic}
If $\calI(\calK)$ is a comparability graph  then $\mis(\calK)=\mhs(\calK)$.
\end{lemma}

From an algorithmic perspective, finding a maximum antichain and a minimum chain covering on a partial ordered set $\calK$ can
be done by computing a minimum vertex cover and a maximum matching on a bipartite graph~(see, e.g.~\cite{FordF1962}).
Since for bipartite graphs a minimum vertex cover can be obtained from a maximum matching in time proportional to the number of edges in the graph~\cite{Schrijver2003}, the following result holds:
\begin{lemma}(See, e.g., \cite{FordF1962}) \label{lema:antichainalg}
The maximum antichain and the minimum chain covering of a partial order $\calK$ can
be found in $O(\matching(v,e))$, where $v$ is the number of
elements in $\calK$, $e$ is the number of pairs of comparable elements in $\calK$,
and  $\matching(v,e)$ is the running time of an algorithm that solves the
bipartite matching problem with $v$ nodes and $e$ edges.
\end{lemma}

\section{An LP based algorithm for \MIS on \dorgs}\label{sec:lp}

In this section we provide an integrality result for the natural linear relaxations of the \MIS problem on \dorgs. 
Consider an arbitrary family of rectangles $\calC$ (not necessarily a \dorg) whose members  have integer vertices in $[n]^2$.
The \defi{fractional clique constrained independent set polytope} associated to $\calC$ is the polyhedron
 $$\QSTAB(\calC) \equiv \{x \in \RR^{\calR}\colon \sum_{R \in \calC:\ q \in R}x_R
\leq 1  \text{ for all $q \in [n]^2$}; x \geq 0\}.$$

The next result follows directly from Lovasz's Perfect Graph Theorem~\cite{Lovasz1972}.
\begin{proposition}\label{prop:lovasz}
Given a family of rectangles $\calC$ with vertices in $[n]^2$, the polytope $\QSTAB(\calC)$ is integral if and only if $\calI(\calC)$ is a perfect graph.
\end{proposition}

For example, this is the case (by Proposition~\ref{prop:cfi-comparability}), when $\calC$ is a c.f.i.~family, as comparability graphs are known to be perfect. Unfortunately, the same does not hold for \dorgs. Figure~\ref{jump-fig7} shows a \dorg $\calR$ for which $\QSTAB(\calR)$ has a non-integral vertex.
Because this polytope is not a full description of the convex hull of characteristic
vectors of independent sets, linear optimization over $\QSTAB(\calR)$ may lead to non-integral solutions.
\begin{figure}[ht]
\centering
\includegraphics[scale=0.8]{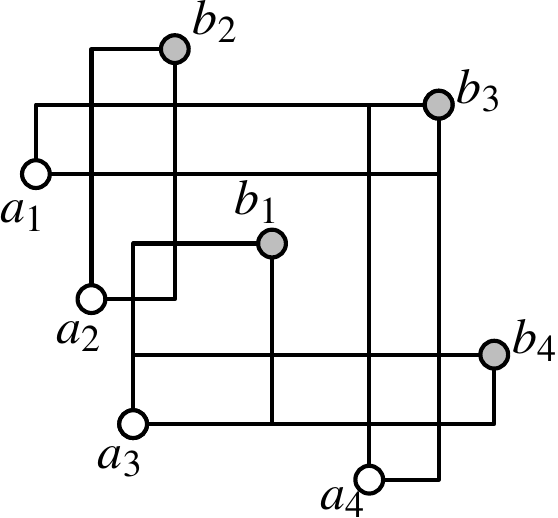}
\caption{A \dorg $\calR(A,B)$ with a non-integral independent set
polytope. The point $x \in \RR^{\calR} $ satisfying $x_R=1/2$ for the five
rectangles shown, and $x_R=0$ for every other rectangle, is a vertex with
non-integral coordinates.}\label{jump-fig7}
\end{figure}

Nevertheless, we will prove that the linear program obtained when we optimize on the all-ones direction:
\begin{align*}
 \mis_\LP(\calR) &=\max\bigg\{\sum_{R \in \calR}x_R: x \in \QSTAB(\calR)
\bigg\},
\end{align*}
has an optimal integral solution. Recalling that $\mis(\calR)$ is the integral version of $\mis_\LP(\calR)$, this result can be stated as follows:

\begin{theorem}\label{thm:opt1}
Let $\calR$ be a non-empty \dorg. There is an optimal integral solution for $\mis_\LP(\calR)$. In other words, $\mis_\LP(\calR)=
\mis(\calR)$.
\end{theorem}

 To prove Theorem~\ref{thm:opt1}, we look at optimal solutions for $\mis_\LP(\calR)$
minimizing the geometric area. More precisely, we consider the modified linear
program $\overline{\mis}_\LP(\calR)$, where $\area(R)$ denotes the geometric
area of a rectangle $R\in \calR$ and $z^*$ denotes the optimal cost of
$\mis_\LP(\calR)$:
\begin{align*}
  \overline{\mis}_\LP(\calR) &=\min\bigg\{\sum_{R \in \calR}\area(R)x_R:
\sum_{R \in \calR}x_R =z^*; x \in \QSTAB(\calR) \bigg\}\enspace.
\end{align*}

Here we are using the function $\area\colon \calR \to \RR$, given by $\area([a_x,b_x]\times[a_y,b_y])=(b_x-a_x)\cdot (b_y-a_y)$, but the argument that
follows works for any function satisfying (1) \emph{Nonnegativity:} $\area(R)\geq 0$, (2) \emph{Strict monotonicity:} If $R \subset S$, then $\area(R)
< \area(S)$, and (3) \emph{Crossing bisubmodularity:} If $R=I\times J$, $S=I'\times J'$ are two corner-intersecting rectangles in $\calR_\downarrow$
then $\area((I \cup I')\times (J \cap J')) + \area((I \cap I')\times (J \cup J')) < \area(R) + \area(S)$.

Let $x^*$ be an arbitrary optimal extreme point of $\overline{\mis}_\LP(\calR)$
and let $\calR_{0}= \{R \in \calR: x^*_R>0 \}$ be the \emph{support} of $x^*$.
Since $x^*$ minimizes weighted area, $\calR_{0}$ is a subset of the collection of inclusion-wise minimal
rectangles $\calR_{\downarrow}$. 

By Lemma~\ref{lema:inclusionwiseminimal}, the only type of corner-intersection that can occur in $\calR_{\downarrow}$ is the one depicted in Figure
\ref{jump-fig9}~(a). Next, we prove that those intersections never arise in $\calR_{0}$.

\begin{proposition}\label{prop:cf}
The family $\calR_0$ is c.f.i.
\end{proposition}
\begin{proof}

Suppose that $R=\Gamma(a,b)$ and $R'=\Gamma(a',b')$ are rectangles in $\calR_{0}$ with corner-intersection as in Figure \ref{jump-fig9}~(a).
 \begin{figure}[ht]
\centering
\includegraphics[scale=1]{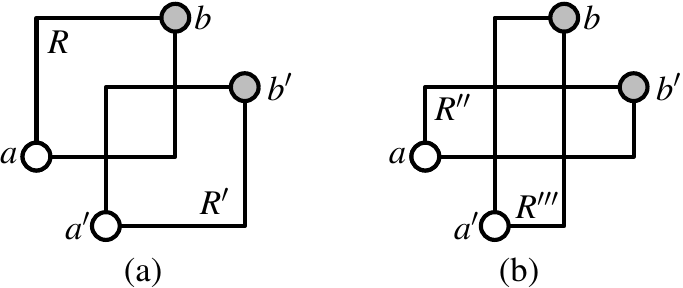}
\caption{Proof of Proposition \ref{prop:cf}.}\label{jump-fig9}
\end{figure}
In Figure \ref{jump-fig9}~(b) we show two
 rectangles, $R''=\Gamma(a,b')$ and $R'''=\Gamma(a',b)$, that also belong to $\calR_{\downarrow}$. We modify $x^*$, first reducing the values of $x^*_R$ and $x^*_{R'}$ by $\lambda\equiv \min\{x^*_R, x^*_{R'}\}$ and then increasing the values of $x^*_{R''}$ and $x^*_{R'''}$ by $\lambda$. Let $\bar{x}^*$ be this modified solution.

Since the weight of the rectangles
covering any point in the plane can only decrease (that
is, $\sum_{S\ni q} x^*_S \geq \sum_{S\ni q} \bar{x}^*_S, \forall q\in
[n]^2$), all the constraints $\sum_{S\ni q} \bar{x}^*_S \leq 1$ must hold. We also have $\sum_{S\in \calR}\bar{x}^*_S =\sum_{S\in \calR} x^*_S =
z^*$, and therefore $\bar{x}^*$ is a feasible solution for $\overline{\mis}_\LP(\calR)$. But the total weighted area of $\bar{x}^*$ is $\lambda\left((\area(R)+\area(R'))-(\area(R'')+\area(R''') \right)>0$ units smaller than the one of the original solution $x^*$, contradicting its optimality.
\end{proof}

The fact that the support of $x^*$ forms a corner-free intersection family is all we need to prove the integrality of $x^*$:

\begin{proposition}\label{prop:integralpoint}
  The point $x^*$ is integral. In particular, $\calR_0$ is a maximum independent set of $\calR$.
\end{proposition}
\begin{proof}
Since $\calR_0$ is the support of $x^*$, the  linear program $\mis_\LP(\calR_0)$
also has $x^*$ as optimal solution.
But $\calI(\calR_0)$ is a comparability graph by Proposition~\ref{prop:cfi-comparability}; therefore, $\QSTAB(\calR_0)$ is integral. To
conclude, we show that $x^*$ is an extreme point of
$\QSTAB(\calR_0)$. Indeed, if this does not hold, then $x^*$ can be written as convex combination of two points in
$\QSTAB(\calR_0) \subseteq \QSTAB(\calR)$, contradicting the fact that $x^*$ is
extreme of $\QSTAB(\calR)$.
\end{proof}
The previous proposition concludes the proof of Theorem~\ref{thm:opt1}. Furthermore, since $\QSTAB(\calR)$ has polynomially many variables and
constraints, we obtain the following algorithmic result.

\begin{theorem}\label{thm:dorg-polynomial}
We can compute a \MIS of a \dorg in polynomial time.
\end{theorem}

\section{A combinatorial algorithm for \MIS and \MHS of a \dorg} \label{sec:combalg}
In this section we devise an algorithm that constructs a maximum independent set
$I^*$ and a minimum hitting set $H^*$ of a given \dorg
$\calR(A,B,\calZ)$. The overall description of our algorithm is as follows. We
first
construct the family $\calRD$ of inclusionwise minimal rectangles of $\calR$.
Using a greedy procedure, we then find a maximal c.f.i.~family of rectangles
$\calK \subseteq \calRD$. Next, we use Lemma~\ref{lema:antichainalg} to
construct an independent set $I^*$ and a hitting set $H_0$ of the same size for
the family $\calK$.  Afterwards, we use a \emph{flipping procedure} to transform
$H_0$ into  a set $H^*$ of the same cardinality  in such a way that $H^*$ is
also a hitting set for $\calRD$, and therefore for $\calR$.
Since $I^*$ and $H^*$ have the same cardinality, they are both optimal.

The following notation will be useful to describe our algorithm.
Given  a rectangle~$R$, let $\vertA(R)$, $\vertB(R)$, $\vertC(R)$ and
$\vertD(R)$ be the bottom-left corner, top-right corner, top-left corner and
bottom-right corners of $R$, respectively. The notation is
chosen so  that if $R$ is a rectangle of the \dorg $\calR(A, B, \calZ)$, then
 $\vertA(R) \in A$ and $\vertB(R)\in B$ are the two defining corners of $R$.

Let $\calRD$ be the set of inclusion-wise minimal rectangles of $\calR$. 
The construction of the c.f.i.~family $\calK$ uses a specific order
in $\calRD$: We say that a rectangle $R$ \emph{precedes} a
rectangle $S$ in \emph{right-top order} if the top-left corner of $R$ is
lexicographically smaller than the top-left corner of $S$, this is,
\begin{itemize}
\item $\vertC(R)_x < \vertC(S)_x$, or
\item $\vertC(R)_x = \vertC(S)_x$ and $\vertC(R)_y< \vertC(S)_y$ .
\end{itemize}

Let $R_1, R_2, \ldots, R_t$ be the rectangles in $\calRD$ sorted in
right-top order. Observe that the top-left corners of these
rectangles are sorted from left to right and, in case of ties,
from bottom to top, hence the notation \emph{right-top} order. In
Figure~\ref{fig:right-top} there is an example of rectangles sorted in this way.

\begin{figure}[ht]
\centering
\includegraphics[scale=0.9]{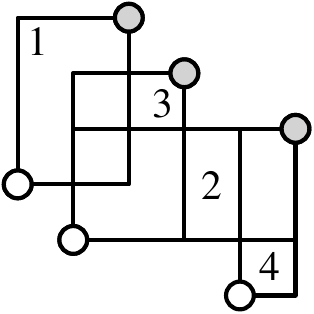}
\caption[The right-top order]{ Some rectangles labeled in right-top
order.}\label{fig:right-top}
\end{figure}

To construct $\calK$, we process the rectangles in the sequence
$(R_i)_{i=1}^t$ one by one. The rectangle $R_i$ being processed is added to~$\calK$ if its addition keeps $\calK$ corner-free. 
The set~$\calK$ obtained by this simple greedy procedure
is a maximal c.f.i.~subfamily of $\calRD$. As stated before, we then use
Lemma~\ref{lema:antichainalg} to compute a maximum independent set $I^*$
and a minimum hitting set $H_0$ for $\calK$. As we will prove later, $I^*$
is actually a maximum independent set of the entire family $\calR$. We can
further assume that the points in $H_0$ have integral coordinates in $[n]^2$,
because every rectangle hit by $(p_x,p_y)$ is also hit by $(\ceil{p_x},\ceil{p_y})$.

The last step of the algorithm consists of a flipping procedure that
essentially moves the points of $H_0$ around so that, in their new position,
they not only hit $\calK$ but the entire family $\calR$. As we discuss in
Section \ref{sec:discussion}, the underlying ideas for this
flipping algorithm are already present in Frank's algorithmic proof of Gy$\ddot{\mathrm{o}}$ri's
min-max result on intervals~\cite{Frank99}.

Consider a set $H \subseteq \ZZ^2$ and two points $p,q \in H$
with $p_x < q_x$, and $p_y < q_y$. By \emph{flipping} $p$ and $q$ in $H$ we
mean to move these two points to the new coordinates $r=(p_x,q_y)$ and
$s=(q_x,p_y)$. More precisely, this means replacing $H$ by $H'=H\setminus\{p,q\}
\cup \{r,s\}$. The following lemma states that flipping points of $H$
that are inside a rectangle of $\calRD$ can only increase the family of
rectangles of $\calRD$ that are hit by $H$.

\begin{lemma}\label{lem:increasing} Let $R$ be a rectangle in $\calRD$ that
contains $p$ and $q$. If
$S \in \calRD$ is hit by $H$ then it is also hit by $H'=H\setminus\{p,q\} \cup
\{r,s\}$.
\end{lemma}
\begin{proof}
Assume by contradiction that $S$ is hit by $H$ but not by $H'$. Assume also that
$S$ is hit by $p$ (the case where $S$ is hit by $q$ is analogous). Since $S$
does not contain $r$ nor $s$, the point $b=\vertB(S)$
must be in the region $[p_x,q_x-1]\times[p_y,q_y-1]$. In particular, $b \in
R\setminus\{\vertA(R),\vertB(R)\}$. By
Lemma~\ref{lema:inclusionwiseminimal}, this
contradicts the inclusion-wise minimality of $R$. \end{proof}

We can now describe the flipping procedure. Let $K_1,\ldots, K_k$ be
the sequence of rectangles in $\calK$ in right-top order, and $H$ be
a set initially equal to $H_0$. In the $j$-th iteration of this
procedure, we find the point $p$ in $H\cap K_j$ of minimum $y$-coordinate
and the point $q$ in $H\cap K_j$ of maximum $x$-coordinate (breaking ties
arbitrarily). Since we maintain $H$ as a hitting set for $\calK$, both points
$p$ and $q$ exist. We then check if both points can be flipped (i.e. if
$p_x<q_x$ and $p_y<q_y$) and in that case, we update $H$ by flipping $p$ and
$q$. The description of our entire algorithm is depicted as
Algorithm~\ref{alg:mis-mhs}.

\begin{algorithm}[!ht]
\caption{for \MIS and \MHS 
of a \dorg $\calR(A,B,\calZ)$.}
\label{alg:mis-mhs}
\begin{algorithmic}[1]
\State{Construct $\calRD$ and sort them as $(R_i)_{i=1}^t$ in right-top order.}
\State{Greedily construct a maximum c.f.i.~family $\calK$ from the sequence
$(R_i)_{i=1}^t$. Let $(K_j)_{j=1}^k$ be the family $\calK$ sorted in right-top
order.}
\State{Find  a
maximum independent set $I^*$ and a minimum hitting set $H$ for~$\calK$.}
\For{$j=1$ \textbf{to} $k$}  \Comment{Flipping procedure starts.}
\State{Let $p$ be a point of minimum $y$-coordinate in $H \cap K_j$ and $q$ be
a point of maximum $x$-coordinate in $H \cap K_j$.}
\If{$p_x < q_x$ and $p_y < q_y$}
\State{$H \gets H\setminus \{p,q\} \cup \{(p_x,q_y), (q_x,p_y)\}$} \Comment{Flip
$p$ and $q$ in $H$.}
\EndIf
\EndFor
\State{Return $I^*$ and $H^* \gets H$.}
\end{algorithmic}
\end{algorithm}

To analyze the correctness of Algorithm~\ref{alg:mis-mhs} we need an auxiliary definition.
Recall that if a rectangle $R_i$ in the sorted sequence $\calRD$ is not included
in
$\calK$ then $R_i$ must have corner-intersection with a rectangle $R_{i'}\in
\calK$ with $i'<i$. The rectangle $R_{i'} \in \calK$ with \emph{largest}
index that has corner-intersection with $R_i$ is denoted as the \emph{witness}
of $R_i$, and written as $R_{i'}=\wit(R_i)$. By Lemma~\ref{lema:inclusionwiseminimal}, the fact that $\wit(R_i)$ precedes
$R_i$ in right-top order implies that $\vertD(\wit(R_i)) \in \tint(R_i)$ and $\vertC(R_i) \in
\tint(\wit(R_i))$.

\begin{lemma}\label{lem:correctness}
After iteration $j$ of the flipping
procedure, the set $H$ hits every rectangle in $\calK$ and every rectangle in
$\calRD\setminus \calK$ with witness in $\{K_l\colon l\leq j\}$.
\end{lemma}
\begin{proof} Observe that the algorithm only flips pairs of points that are
inside a given rectangle $K_j$ of~$\calRD$. By Lemma~\ref{lem:increasing}, the
collection of rectangles hit by $H$ increases in every iteration. So we only
need to prove that after iteration $j$ all the rectangles witnessed by $K_j$
are hit. We do this by induction.

If every rectangle witnessed by $K_j$ is hit at the end of iteration
$j-1$, we are done. Suppose then, that at then end of iteration $j-1$ there is at least
one rectangle $R$ in $\calRD\setminus \calK$ witnessed by $K_j$ that has not
been hit. Let $a=\vertA(K_j), b=\vertB(K_j), a'=\vertA(R)$ and $b'=\vertB(R)$ so
that $K_j=\Gamma(a,b)$ and $R=\Gamma(a',b')$. Since $K_j$ precedes $R$ in
right-top order and they have corner-intersection, the relative position of
both rectangles must be as depicted in Figure~\ref{fig:RK}.

\begin{figure}[ht]
\centering
\includegraphics[scale=0.7]{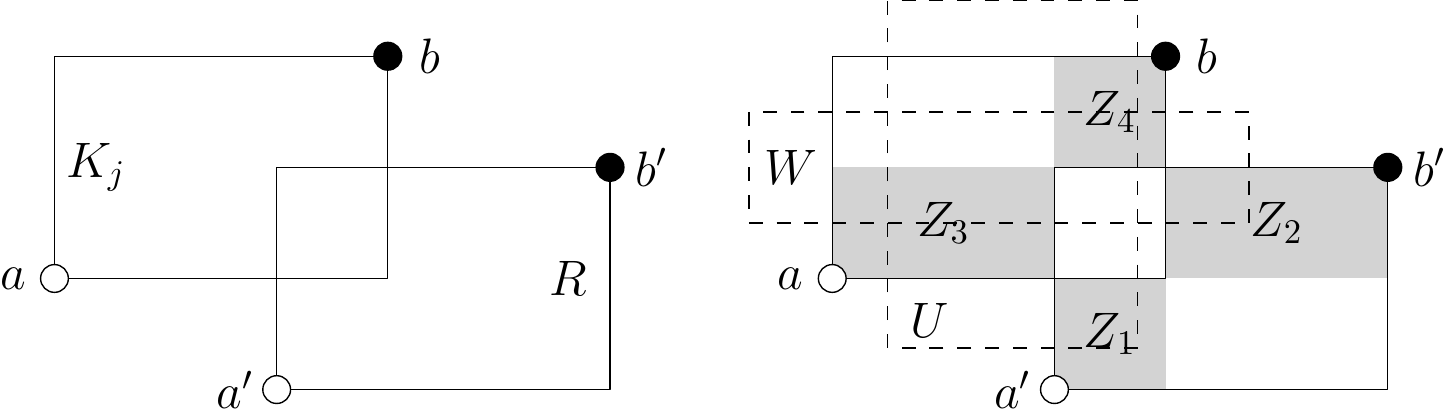}
\caption{On the left, the two rectangles $K_j = \Gamma(a,b)$ and
$R=\Gamma(a',b')$ with $\wit(R)=K_j$. On the right, auxiliary
constructions for the proof of Lemma~\ref{lem:correctness}}\label{fig:RK}
\end{figure}

Observe that the rectangles $S=\Gamma(a',b)$ and $T=\Gamma(a,b')$ are in
$\calRD$. We claim that $S \in
\calK$ and that $T$ is already hit by $H$.

Let us prove the first claim. Assume this is not the case; then $S
\in \calRD\setminus \calK$ must have a witness $U=\wit(S)$ in $\calK$. Let
$d=\vertD(U)$ be the
bottom-right corner of $U$. Since $S$ and $U$ have corner-intersection and $U$
precedes~$S$, the point $d$ is in the interior of $S$. Furthermore, since $U$ and $K_j$
are both in $\calK$,  $d$ cannot be in the interior of $K_j$. We conclude that
$d \in
Z_1 := \tint(S)\setminus \tint(K_j) =(a'_x,b_x)\times (a'_y, a_y]$. See
Figure~\ref{fig:RK} for reference. But since $U$ contains the
top-left corner of $S$ in its interior (because $U$ and $S$ have
corner-intersection) and $U$ does not contain $a$ on its interior (as otherwise $U$ would not be
inclusionwise minimal), we conclude that
$K_j$ precedes $U$ in right-top order. But then, $U$ is a rectangle in $\calK$
having corner-intersection with $R$ and appearing after $K_j$ in right-top
order.
This contradicts the definition of $K_j$ as witness of $R$, and concludes the
proof of the first claim.

Let us prove the second claim. If $T=\Gamma(a,b')$ is in $\calK$ then we are
done as $H$ is a hitting set for $\calK$. So assume that $T \in \calRD
\setminus \calK$. Let $W=\wit(T)$ and $d'=\vertD(W)$. Since $W$ and $T$ have
corner-intersection and $W$ precedes $T$, $d'$ is in the interior of $T$.
Furthermore, since both $W$ and $K_j$ are in $\calK$, $d'$ is not in the
interior of $K_j$. We conclude that $d'\in Z_2:=\tint(T)\setminus \tint(K_j) =
[b_x, b'_x)\times (a_y, b'_y)$.
Since $W$ contains the top-left corner of $T$ in its interior (because $W$ and
$T$ have corner-intersection), we deduce that $W$ precedes $K_j$ in right-top
order. But then $W=K_{j'}$ for $j'<j$ and so, by induction
hypothesis, rectangle $T$ was hit at the end of iteration $j'$. This concludes
the proof of the second claim.

Recall that $R$ is  not hit by $H$ at the end of iteration $j-1$. Since $T$ is
hit by $H$ the set $H \cap T \setminus R$ must be nonempty. In particular, the
point $p$ chosen by the algorithm, of minimum $y$-coordinate in $K_j \cap H$
must be in the zone $Z_3 := [a_x, a'_x) \times [a_y, b'_y]$. Similarly, since $S
\in \calK$, it must be hit by $H$ and so the set $H \cap S \setminus R$ is
nonempty. Therefore, the point $q$ chosen by the algorithm of maximum
$x$-coordinate in $K_j \cap H$ is in the zone $Z_4 := [a'_x, b_x]\times
(b'_y,b_y]$. In particular, $p_x < q_x$, $p_y < q_y$ and the point $s=(q_x,p_y)$
is in $R$.
We conclude that after flipping $p$ and $q$ in $H$, the rectangle $R$ is
hit. \end{proof}

Thanks to the previous lemma we obtain our main combinatorial result.

\begin{theorem}\label{teorem:teo1} Algorithm~\ref{alg:mis-mhs} returns an independent set $I^*$ and a hitting set $H^*$ of $\calR(A,B,\calZ)$ of the
same cardinality. In
particular, they are both optimal and
$\mis(\calR)=\mhs(\calR)$.
\end{theorem}
\begin{proof}
The fact that $I^*$ is an independent set of $\calR$ follows by construction
and the fact that $\calK \subseteq \calR$. By Lemma~\ref{lem:correctness},
the set $H^*$ returned by the algorithm hits all the elements in $\calRD$.
Since every rectangle in $\calR$ contains a rectangle in $\calRD$, $H^*$ is
also a hitting set of $\calR$. Finally, since $H^*$ was constructed from
$H_0$ via a sequence of flips which preserve cardinality and since
$|H_0|=|I^*|$ by Lemma~\ref{lema:minmaxbasic}, we obtain that
$|I^*|=|H^*|$. As every hitting set has cardinality as least as large as every
independent set, we conclude that both are optimal.
\end{proof}

It is quite simple to  give a polynomial time implementation of
Algorithm~\ref{alg:mis-mhs}. In Subsection~\ref{sub:Implementation}, we discuss an
efficient implementation that runs in
$O(\matching(n\log n, n^2) + n^2\log^2 n)$ time, where
$n=|A\cup B|$. Here we are assuming that testing containment in $\calZ$ can be done in unit time. Otherwise, we need additional $O(n^2T(\calZ))$ time,
where $T(\calZ)$ is the time for testing containment in~$\calZ$.

\subsection{Implementation of the combinatorial algorithm}\label{sub:Implementation}

To implement our algorithm it will be useful to have access to a data structure
for dynamic orthogonal range queries. That is, a structure to store a
dynamic collection of points $P$ in the plane, supporting insertions, deletions
and queries of the following type: given an axis-parallel rectangle $Q$, is
there a point $p$ in $Q\cap P$? And if there is one, report any.

There are many data structures we can use. For our purposes, it would be enough
to use any structure in which each operation takes polylogarithmic (or even
subpolynomial) time. For concreteness, we use the
following result of Willard and Lueker~\cite{WillardL1985}, specialized to
two-dimensional Euclidean space.

\begin{theorem}[(Willard and Lueker)]
 There is a \emph{point data structure} for orthogonal range queries in the plane on
$n$ points supporting insertion, deletion and queries in time $O(\log^2 n)$,
and using space $O(n\log n)$.
\end{theorem}

Using this data structure, we can easily implement
Algorithm~\ref{alg:mis-mhs}. Indeed, let us first see how to construct $\calRD$.

Start by inserting the points in $A$ and $B$ to the point data structure and
creating an empty list $\calRD$. For each point $a$ in $A$ sorted from left to
right and each point $b$ in $B$ from bottom to top check if $\Gamma(a,b)$ is a
rectangle in $\calR$ (i.e., if $a \leq_{\RR^2} b$ and if $\Gamma(a,b)\subseteq \calZ$) and if it has no points of $A\cup B$ in
its interior using the data structure.  If both conditions hold, add
$\Gamma(a,b)$ to the end of $\calRD$. Note that by going through $A\times B$ in
this order, the list $\calRD$ is sorted in right-top order. It is easy
to see that the entire procedure takes time $O(n^2 \log^2 n + n^2T(\calZ))$.

Constructing $\calK$ is similar. Start by creating an empty list $\calK$ and a
point data structure $\vertD$ containing the bottom-right corners of all rectangles
in~$\calRD$. Then, go through the ordered list $\calRD$ once again and add the current
rectangle $R\in \calRD$ to the end of $\calK$ if $R$ does not contain a point of $\vertD$ in its interior. It is easy to see that we obtain the
sorted c.f.i.~family $\calK$ described in our algorithm in this way and
that the entire procedure takes time $O(|\calRD|\log^2 (|\calK|))=O(n^2 \log^2
(|\calK|))$.

Afterwards, construct the intersection graph $\calI(\calK)$ in
$O(|\calK^2|)$ time and then use Lemma~\ref{lema:antichainalg} to get a
maximum independent set $I^*$ and a minimum hitting set\footnote{More
precisely, the algorithm gives a chain partition of $\calI(\calK)$. This can
be transformed into a hitting set of $\calK$ by selecting on each chain returned
the bottom-left point of the mutual intersection of all rectangles in the
chain. The extra processing time needed is dominated by $O(|\calK|)$.} $H_0$ of
$\calK$ in time $O(\matching(|\calK|, |E(I(\calK))|)$.

To implement the flipping procedure, we initialize a point data structure
containing $H$ at every moment. Note that $|H|\leq n$ as $A \cup B$ is itself a
hitting set of $\calR$. In each of the $|\calK|$
iterations of the flipping procedure we need to find the lowest point and the
rightmost point of a range query. This can be done using binary search and the
query operation of the data structure losing an extra logarithmic factor, i.e.,
in time $O(\log^3 n)$. To flip two points of $H$, we perform
two deletions and two insertions in time $O(\log^2 n)$. The entire flipping
procedure takes time $O(|\calK|\log^3 n)$.

In Subsection~\ref{sub:bounds} we  prove that $|\calK|=O(n\log n)$ and
$|E(\calI(\calK))|=O(n^2)$. By using these bounds and the previous discussion, we
conclude that Algorithm~\ref{alg:mis-mhs} can be implemented to run in time
$O(\matching(n\log n, n^2) + n^2\log^2 n + n^2T(\calZ)).$

Hopcroft and Karp's~\cite{hopcroft1973n} algorithm for maximum matching on a bipartite graph
with $v$ vertices and $e$ edges runs in time $O(e\sqrt{v})$. Specializing this
to $v=O(n\log
n)$ and $e=O(n^2)$, $\matching(n\log n, n^2)$ is time $O(n^{2.5}\sqrt{\log n})$. On the other hand,
Mucha and Sankowski~\cite{mucha2004maximum}  have devised a randomized algorithm that returns
with high probability a maximum matching of a bipartite graph in time
$O(v^\omega)$,  where $\omega$ is the exponent for square matrix multiplication.
The current best upper bound for $\omega$ is approximately $2.3727$ by
Williams~\cite{Williams12}. From this discussion, we
obtain the following result.

\begin{theorem} We can implement Algorithm~\ref{alg:mis-mhs}  to
run in time $$O\left(\matching(n\log n, n^2) + n^2\log^2 n +n^2T(\calZ)\right),$$
where
$n=|A\cup B|$ and $T(\calZ)$ is the time needed to test if a rectangle is contained in $\calZ$. Using Hopcroft and Karp's implementation, this
is $$O(n^{2.5}\sqrt{\log n} +n^2T(\calZ)).$$ Using Mucha and Sankowski's randomized
algorithm, the running time can be reduced to $$O((n\log n)^\omega+n^2T(\calZ)),$$ where
$\omega <2.3727$ is the exponent for square matrix multiplication.
\end{theorem}

\subsection{Bounds for $\calK$.}\label{sub:bounds}

The bounds we prove on this section are valid for every corner-free
intersection subfamily of $\calRD$.
Given one such family $\calK$, define  its lower-left corner set
$\vertA(\calK)=\{\vertA(K)\colon K \in \calK\}$ and its upper-right corner set
$\vertB(\calK)=\{\vertA(K)\colon K \in \calK\}$. We start with a simple result.

\begin{lemma}\label{lem:verticalline}
  If all the rectangles of $\calK$ intersect a fixed vertical (or horizontal)
line $\lambda$, then $|\calK| \leq \abs{\vertA(\calK)} + \abs{\vertB(\calK)}
- 1.$
\end{lemma}

\begin{proof}
  Project all the rectangles in $\calK$ onto the line $\lambda$ to obtain a
collection of intervals in the line. Since $\calK$ is a c.f.i.~family, the
collection of intervals forms a \emph{laminar} family: if two intervals
intersect, then one is contained in the other. Let $X$ be the collection of
extreme points of the intervals. Since by assumption no two points in $A
\cup B$ share coordinates, $|X|=\abs{\vertA(\calK)} + \abs{\vertB(\calK)}$ and
furthermore, every interval is a non-singleton interval. To conclude the proof
of the lemma we use the following known fact, which can be proved by
induction: every laminar family of non-singleton intervals with
extreme points in $X$ has cardinality at most $\abs{X}-1$.
\end{proof}

Now we consider the situation where the family $\calK$ is not necessarily
stabbed by a single line.

\begin{lemma}\label{lem:verticallines}
Let $\mathcal{L}=\{\lambda_1,\ldots, \lambda_r\}$ be a collection of vertical
lines that intersect all the rectangles in $\calK$,
then $|\calK| \leq n(1+\lfloor \log_2(r) \rfloor) - r.$
\end{lemma}
\begin{proof}
Assume that $\lambda_1, \ldots, \lambda_r$ is sorted from left to right.
Consider the following collections of vertical lines:
\begin{align*}
  \mathcal{L}_0 &= \{\lambda_{2k+1}\colon k \geq 0\} =
\{\lambda_1,\lambda_3,\lambda_5,\lambda_7, \lambda_9\ldots\}.\\
  \mathcal{L}_1 &= \{\lambda_{4k+2}\colon k \geq 0\} =
\{\lambda_2,\lambda_6,\lambda_{10},\lambda_{14},\ldots\}.\\
  \mathcal{L}_2 &= \{\lambda_{8k+4}\colon k \geq 0\} =
\{\lambda_4,\lambda_{12},\lambda_{20},\lambda_{28},\ldots\}.\\
  &\cdots\\
  \mathcal{L}_t &= \{\lambda_{2^t (2k + 1)}\colon k \geq 0\}.
\end{align*}
The collections $\mathcal{L}_0, \dots,\mathcal{L}_{\lfloor \log_2(r)\rfloor }$
form a partition of $\mathcal{L}$. Include each rectangle $R$ in $\calK$ into
the set $\calK_t$, $0\leq t \leq \lfloor
\log_2(r)\rfloor$, of largest index such that $\mathcal{L}_t$ contains a
vertical line that intersects $R$. See Figure~\ref{fig:verticalpartition} for an
example illustrating this construction.

\begin{figure}[h!]
  \centering \includegraphics{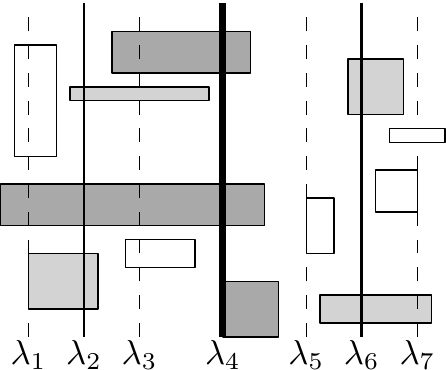}
  \caption{The dotted lines, thin lines and thick lines correspond to
$\mathcal{L}_0$, $\mathcal{L}_1$ and $\mathcal{L}_2$ respectively. White,
light-gray and dark-gray rectangles are assigned to $\calK_0$, $\calK_1$ and
$\calK_2$ respectively.}
  \label{fig:verticalpartition}
\end{figure}

Fix an index $t$. Every rectangle in $\calK_t$ intersect a unique line in
$\mathcal{L}_t$ (if it intersects two or more, then it would also intersect a
line in $\mathcal{L}_{t+1}$). For a given line $\lambda \in \mathcal{L}_t$, let
$\calK_{\lambda}$ be the family of rectangles in $\calK_t$ intersecting
$\lambda$. By Lemma~\ref{lem:verticalline}, the number of rectangles in
$\calK_\lambda$ is at most $\abs{\vertA(\calK_\lambda)} +
\abs{\vertB(\calK_\lambda)} -1$. Every point $a$ in $\vertA(\calK)$  belongs to
exactly one set in $\{\vertA(\calK_\lambda): \lambda \in \mathcal{L}_t\}$. It
belongs to the one corresponding to the first line
$\lambda$ is the first line on or to the right of $a$. Therefore, $\sum_{\lambda
\in \mathcal{L}_t} \abs{\vertA(\calK_\lambda)} \leq \abs{A}$, and similarly,
$\sum_{\lambda \in \mathcal{L}_t} \abs{\vertB(\calK_\lambda)} \leq \abs{B}$.
Altogether we get $  \abs{\calK_t} = \sum_{\lambda \in \mathcal{L}_t}
\abs{\calK_\lambda} \leq \abs{\vertA(\calK)} + \abs{\vertB(\calK)} -
\abs{\mathcal{L}_t} \leq \abs{\vertA(\calK)} + \abs{\vertB(\calK)}$.  Summing over $t$ we obtain $
  \abs{\calK} = \sum_{t=0}^{\lfloor \log_2(r) \rfloor} \abs{\calK_\lambda} \leq
(1+\lfloor \log_2(r) \rfloor)(\abs{\vertA(\calK)} + \abs{\vertB(\calK)})$.\qedhere
\end{proof}

If $\calK$ is the c.f.i. family obtained with Algorithm \ref{alg:mis-mhs}, then $n$  vertical lines are enough to intersect all rectangles in $\calK$. Therefore, $|\calK|=O(n\log n)$. Let us now bound the number
 of edges of the intersection graph $\calI(\calK)$.

\begin{lemma}
  $|E(\calI(\calK))| = O(n^2).$
\end{lemma}
\begin{proof}
Let $\ell=|E(\calI(\calK))|$. Let also $\Lambda(n)$ be the maximum possible value
of $\ell$ as a function of $n$. Recall that the vertices of all
rectangles in $\calK$ are in the grid $[n]^2$. Consider the vertical
line $\lambda=\{(x,\lfloor n/2 \rfloor)\colon x \in \RR\}$ that divides the grid
in two roughly equal parts. Count the edges in $\calI(\calK)$ as follows.
Let $E_1$ be the edges connecting pairs of rectangles that are totally to the
left of $\lambda$, $E_2$ be the edges connecting pairs of rectangles that are
totally to the right of $\lambda$, and $E_3$ be the remaining edges. We
trivially get that $|E_1| \leq \Lambda(\lfloor n/2 \rfloor)$ and  $|E_2| \leq
\Lambda(\lceil n/2 \rceil)$. We bound the value of $|E_3|$ in a different way.

Let $\calK_\lambda$ be the rectangles intersecting the vertical line $\lambda$.
Then, $E_3$ is exactly the collection of edges in $\calI(\calK)$ having an
endpoint in $\calK_\lambda$. By Lemma~\ref{lem:verticalline},
$\abs{\calK_\lambda} \leq n$. Now we bound the degree of each element of
$\calK_\lambda$ in $\calI(\calK)$. Consider one rectangle $R=\Gamma(a,b) \in
\calK_\lambda$. Every rectangle intersecting $R$ must intersect one of the four
lines defined by its sides. By using again Lemma~\ref{lem:verticalline}, we
conclude that the degree of $R$ in $\calI(\calK)$ is at most $4n$. Therefore, the
number of edges having an endpoint in $\calK_\lambda$ is at most $4n^2$.

We conclude that $\Lambda(n)$ satisfies the recurrence $
  \Lambda(n) \leq \Lambda(\lfloor n/2 \rfloor) + \Lambda(\lceil n/2 \rceil) +
4n^2,$ from which, $\Lambda(n) = O(n^2)$.\qedhere
\end{proof}

To finish this section, we note that it is is very easy to construct a family
$\calK$, for which $|E(\calI(\calK))| = \Omega(n^2)$. See
Figure~\ref{fig:omegan2} for an example. In this example, $\calK$ was obtained
using the greedy procedure. Also, if we let $n'=|A|=|B|=n/2$, then
$|\calK|=2n'+2$ and $|E(\calI(\calK))|=(n')^2+4 = \Omega(n^2)$.

\begin{figure}[h!]
  \centering \scalebox{0.8}{\includegraphics{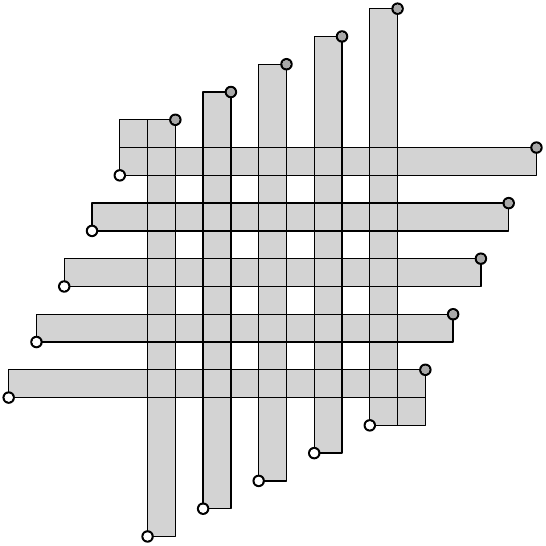}}
  \caption{Example with $|E(\calI(\calK))| = \Omega(n^2)$.}
  \label{fig:omegan2}
\end{figure}

On the other hand, it is not clear if there are examples achieving
$\abs{\calK}=\Omega(n\log n)$. In particular, we give the following conjecture
\begin{conjecture}\label{conj:linear}
  For any c.f.i.~family $\calK$ with vertices in $[n]^2$, $\abs{\calK}=O(n)$.
\end{conjecture}

If this conjecture holds, it is easy to get better bounds for the running time
of Algorithm~\ref{alg:mis-mhs}.

\section{\WMIS of \dorgs and maximum weight cross-free matching}\label{sec:mwbp}

We now  consider the problem of finding a maximum weight independent set of the rectangles in a \dorg $\calR(A, B, \calZ)$ with weights $\{w_R\}_{R\in
\calR}$. This problem is significantly harder than the unweighted counterpart.
\begin{theorem}\label{thm:NPhard}
The maximum weight independent set problem is $\NP$-hard for unrestricted \dorgs even for weights in $\{0,1\}$.
\end{theorem}
\begin{proof}
We reduce from the maximum independent set of rectangles problem, which is $\NP$-hard even if the vertices of the rectangles are all
distinct~\cite{FowlerPT1981}. Given an instance $\mathcal{I}$ with the previous property, let $A$ (resp.~$B$) be the set of lower-left (resp.~upper-right) corners of rectangles
in $\mathcal{I}$. Note that $\mathcal{I} \subseteq \calR(A,B)$ so we can find a maximum independent set of $\mathcal{I}$ by finding
a maximum weight independent set in $\calR$, where we give unit weight to each rectangle $R \in \mathcal{I}$, and zero weight to every other rectangle.
\end{proof}

Despite this negative result, we can still provide efficient algorithms for some interesting subclasses of \dorgs graphs $G=(A\cup B, \calR)$. Since
independent sets of $\calR$ are in one to one correspondence with cross-free matchings of $G$, the problem we are studying in this section is
equivalent to finding a \emph{maximum weight cross-free matching} of $G$. 

\subsection{A hierarchy of subclasses of \dorgs}\label{sec:hierarchy}
We recall the following definitions of a nested family of bipartite graph classes. We keep the notation $G=(A\cup B,\calR)$ since they are all \dorg
graphs.

A \emph{bipartite permutation graph} (or bipartite 2-dimensional graph) is the comparability graph of a two dimensional partially ordered set of height~2, where $A$ is the set of minimal elements and $B$ is the complement of this set. For our purposes, a \emph{two dimensional partially ordered set} is simply a collection of points in $\ZZ^2$ with the relation $p \leq_{\ZZ^2} q$ if $p_x \leq q_x$ and $p_y \leq q_y$. 

A \emph{convex bipartite graph} $G=(A\cup B,\calR)$ is a bipartite graph admitting a labeling $\{a_1, \ldots, a_k\}$ of $A$ so that the neighborhood of each $b\in B$ is a set of consecutive elements of $A$. A \emph{biconvex graph} is a convex bipartite graph for which there is also a labeling for $B=\{b_1, \ldots, b_l\}$ so that the neighborhood of each $a\in A$ is consecutive in $B$.

In an \emph{interval bigraph}, each vertex $v\in A\cup B$ is associated to a real closed interval $I_v$ (w.l.o.g.~with integral extremes) so that
$a\in A$ and $b\in B$ are adjacent if and only if $I_a \cap I_b\neq \emptyset$. 

A \emph{two directional orthogonal ray graph} (2dorg) is a bipartite graph on $A\cup B$ where each vertex $v$ is associated to a point $(v_x,v_y)\in \ZZ^2$, so that $a\in A$ and  $b\in B$ are connected if and only the rays $[a_x,\infty)\times\{a_y\}$ and $\{b_x\}\times(-\infty,b_y]$ intersect each other. Since this condition is equivalent to $\Gamma(a,b) \in \calR(A,B)$, two directional orthogonal ray graphs are exactly the unrestricted \dorgs graphs.

It is known that the following strict inclusions hold for these classes~\cite{BrandstadtLS1999,ShresthaTU2010}: bipartite permutation $\subset$
biconvex $\subset$ convex $\subset$ interval bigraph $\subset$ 2dorg. There are also polynomial time algorithms to recognize if a graph belongs to any
of these classes~\cite{ShresthaTU2010,spinrad2003efficient,muller1997recognizing}. On the  other hand, the problem of recognizing general (restricted) \dorg graphs is open.

We give a simple geometrical interpretation of some of the classes presented above as unrestricted \dorgs. 
The equivalence between the definitions below and the ones above are simple so the proof of equivalence is omitted. 
Bipartite permutation graphs are simply the unrestricted \dorg graphs $G=(A\cup B, \calR)$ where no two points in the 
same color class are comparable under $\leq_{\RR^2}$. 
Let $L = \{(x,-x)\colon x \in \ZZ\}$ be the integer points of the diagonal 
line $y=-x$. Similarly, let $L^+ = \{(x,y) \in \ZZ^2\colon y \geq -x\}$ and $L^- = \{(x,y)\in \ZZ^2\colon y \leq -x\}$ be the points weakly above and
weakly below $L$. Convex bipartite graphs are those unrestricted \dorg graphs $G=(A\cup B, \calR)$ where $A \subset L$ and $B\subset L^+$. Interval
bigraphs are those unrestricted \dorg graphs where $A \subset L^-$ and $B\subset L^+$. 

In the following subsections we investigate the maximum weight cross-free matching problem on these subclasses.

\subsection{Bipartite permutation graphs}
Every bipartite permutation graph $G$ is the graph representation of a \dorg $\calR(A,B)$ such that no two points in the same color class ($A$ or
$B$) are comparable under $\leq_{\RR^2}$. In this case we say that $\calR(A,B)$ is a \emph{bipartite permutation} \dorg and we can define a partial
order $\searrow$ on its rectangles. We say that $R \searrow S$ if $R$ and $S$ are disjoint and at least one of $R_x < S_x$ and $R_y > S_y$ holds. It
is not hard to verify (see, e.g., Brandst{\"a}dt~\cite{Branstadt1989}) that $(\calR,\searrow)$ is a partial order whose comparability graph is the
complement of $\calI(\calR)$. 
 In what follows, we use this fact to devise polynomial time algorithms for the maximum weight independent set of bipartite permutation \dorgs.

Observe that $\calI(\calR)$ is a perfect graph (because so is its complement~\cite{Lovasz1972}); therefore, using Proposition~\ref{prop:lovasz} we can
compute a maximum weight independent set of $\calR$ in polynomial time by finding an optimal vertex of
\begin{align*}
  \mis_\LP(\calR,w) &\equiv \max\bigg\{ \sum_{R \in \calR}w_Rx_R\colon \sum_{R:\ q \in R}x_R \leq 1, \text{$q \in [n]^2$}; x\geq 0 \bigg\}\enspace.
\end{align*}

We can also do this combinatorially using that the maximum weight independent sets in $\calR$ are exactly the maximum weight chains on the weight
partially ordered set $(\calR,\searrow)$.

\begin{theorem}\label{thm:algobip}
There is an $O(n^2)$ algorithm for the maximum weight cross-free matching (and hence, for \WMIS) of bipartite permutation graphs. 
\end{theorem}
\begin{proof}
For simplicity, let us assume that all the weights are different. Our algorithm exploits the geometric structure  of the independent sets in $\calR$.
Since $A$ and $B$ are antichains of $\leq_{\RR^2}$, the condition $R \searrow S$ implies that $\vertA(R)_x < \vertA(S)_x$ and $\vertB(R)_y >
\vertB(S)_y$. Let $\calR^* \subset \calR$ be any maximum weight independent set and let $R \searrow S \searrow T$ be three consecutive rectangles in
$\calR^*$. We can extract the following information about $S$ (see Figure \ref{jump-fig28} for the first two scenarios, the other two are analogous):

\begin{itemize}
\item{\textbf{Down-right scenario:}}\\ If $R_y > S_y$ and $S_x < T_x$, then (i) $S$ is the heaviest rectangle with corner $\vertB(S)$. In particular, $S$ is determined by $\vertB(S)$.
\item{\textbf{Down-down scenario:}}\\ If $R_y > S_y$ and $S_y > T_y$, then (ii) $S$ is the heaviest rectangle below $R$ with corner $\vertA(S)$. In particular, $S$ is determined by $\vertA(R)_y$ and $\vertA(S)$.
\item{\textbf{Right-down scenario:}}\\ If $R_x < S_x$ and $S_y > T_y$, then (iii) $S$ is the heaviest rectangle with corner $\vertA(S)$. In particular, $S$ is determined by $\vertA(S)$.
\item{\textbf{Right-right scenario:}}\\ If $R_x < S_x$ and $S_x < T_x$, then (iv) $S$ is the  heaviest rectangle to the right of $R$ with corner $\vertB(S)$. In particular, $S$ is determined by $\vertB(S)$ and $\vertB(R)_x$.
\end{itemize}

 \begin{figure}[ht]
\centering
\includegraphics[scale=0.8]{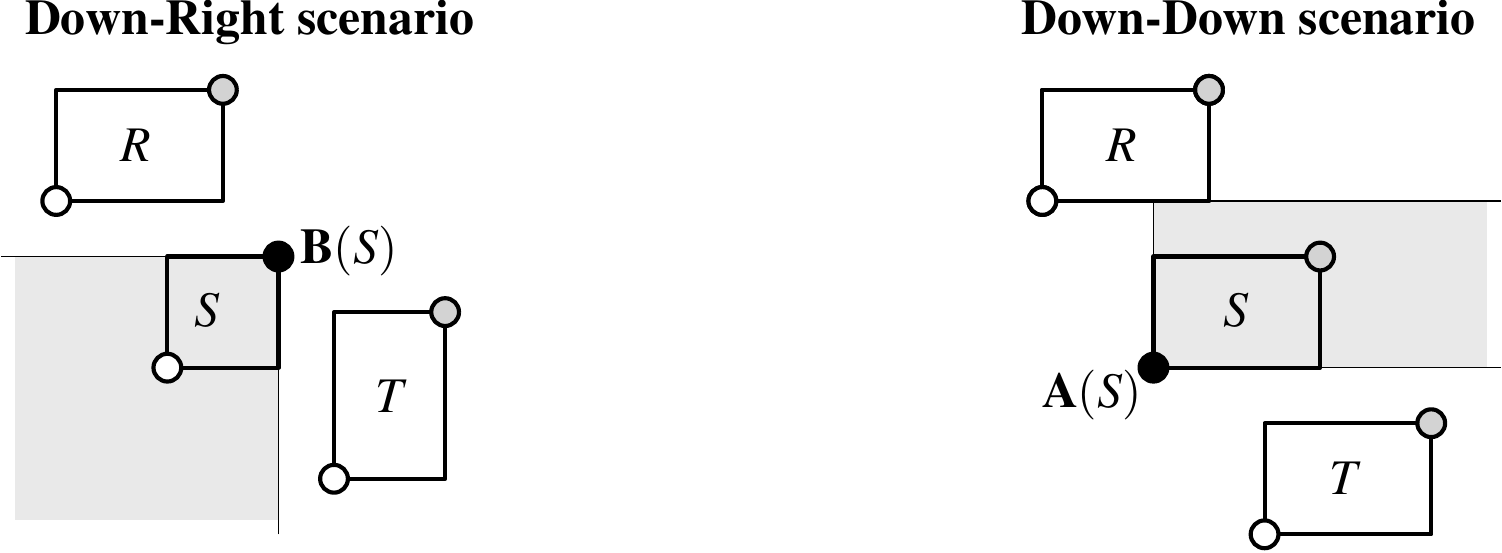}
\caption{Proof of Theorem \ref{thm:algobip}.}\label{jump-fig28}
\end{figure}

 For $R\in \calR$, let $V(R)$ be the maximum weight of a path in $(\cal{R},\searrow)$ that starts with $R$. For $a\in A$,  let $V_{\downarrow}(a)$  be the maximum weight of a path using only rectangles below $a$. Similarly, for $b\in B$,  let $V_{\rightarrow}(b)$ be the maximum weight of a path using only rectangles to the right of~$b$.

Clearly, $V_{\downarrow }(a) = \max\left\{V(S)\colon \text{$S$  rectangle below $a$} \right\}$, but from the decomposition into scenarios we can restrict $S$ to be a rectangle below $a$ satisfying properties (i) or (ii). Using this idea, we define:

\begin{compactitem}
\item $S(a)=$ the heaviest rectangle with $\vertA(S)=a$,
\item $S(a,a')=$ the heaviest rectangle below $a$ with $\vertA(S)=a'$,
\item $T(b)=$ the heaviest rectangle with $\vertB(S)=b$,
\item $T(b,b')=$ the heaviest rectangle to the right of $b$ with $\vertB(S)=b'$,
\item $\calS_a^{\text{(i)}}=\{T(b)\colon \text{$b$ below $a$}\}$,
\item $\calS_a^{\text{(ii)}}=\{S(a,a')\colon \text{$a'$ below $a$}\}$,
\item $\calS_b^{\text{(iii)}}=\{S(a)\colon \text{$a$ to the right of $b$}\}$,
\item $\calS_b^{\text{(iv)}}=\{T(b,b')\colon \text{$b'$ to the right of $b$}\}$,
\end{compactitem}
%
and compute the recursion as follows:
\begin{align}\label{recursion}
V(R) &= \max\left\{V_{\downarrow}(\vertA(R)), V_{\rightarrow}(\vertB(R))\right\}+ w_R, \notag\\
V_{\downarrow }(a) &= \max\left\{V(S)\colon S \in \calS^{\text{(i)}}_a \cup \calS^{\text{(ii)}}_a  \right\}\!,\\
V_{\rightarrow}(b)&= \max\left\{V(S)\colon S  \in \calS^{\text{(iii)}}_b \cup \calS^{\text{(iv)}}_b  \right\}\!.\notag
\end{align}

It is easy to precompute all sets in $\{\calS^{\text{(i)}}_a, a\in A\}$ in $O(n^2)$ time. We can also precompute
$\calS^{\text{(ii)}}_a$ for all $a\in A$ in  $O(n^2)$ time: we fix $a'\in A$, and then traverse the points $a\in A$ from bottom to top, finding
$S(a,a')$, for all $a\in A$ in $O(n)$ time. Iterating now on $a'\in A$, we determine all the sets  $\calS^{\text{(ii)}}_a$ in $O(n^2)$ time. 
After this preprocessing, each rectangle in $\calS^{\text{(i)}}_a$ and $\calS^{\text{(ii)}}_a$ can be accessed in $O(1)$ time. The same holds for
$\calS^{\text{(iii)}}_b$ and $\calS^{\text{(iv)}}_b$, via a similar argument. Since the cardinality of each of these sets is $O(n)$,  and there are
$O(n)$ values $V_{\downarrow }(a)$ and $V_{\rightarrow}(b)$ to compute, the complete recursion for these terms can be evaluated in $O(n^2)$. Finally,
the maximum weight independent set can be found by computing $\max_R\{V(R)\}$ in $O(n^2)$ time and then backtracking the recursion.

If there are repeated weights, we break ties in Properties (i) and (ii) by choosing the rectangle~$S$ of smallest height and  we break ties in Properties (iii) and (iv) by choosing the rectangle $S$ of smallest width. This does not affect the asymptotic running time.
\end{proof}

We can further improve the running time of the previous algorithm when the weights are in $\{0,1\}$ and a suitable description of the input is given.
\begin{theorem}\label{thm4623}
We can compute a \WMIS of a permutation \dorg in $O(n)$ time when the input satisfies the following conditions:
\begin{enumerate}
\item the input graph is given by a biadjacency matrix $M\in \{0,1\}^{\abs{A}\times \abs{B}}$ where the points of $A$ and $B$ are sorted according to
the $x$-coordinate, and where we can access the first and last 1 of every row and column in $O(1)$ time.
\item the weights are given by a matrix $M'\in \{0,1\}^{\abs{A}\times \abs{B}}$ where the points of $A$ and $B$ are sorted according to the
$x$-coordinate, and where we can access the first and last 1 of every row and column in $O(1)$ time.
\end{enumerate}
\end{theorem}
\begin{proof}
Our algorithm uses a simplified version of the algorithm for arbitrary weights introduced in Theorem \ref{thm:algobip}. 
We assume that all the points in $A \cup B$ are defining corners of at least one rectangle with nonzero weight, and therefore both $M$ and $M'$ have
no zero rows or columns. Call a rectangle full if its weight is 1 and void if its weight is 0.
 
Because of our tie-breaking rule, for each $b\in B$ and $a\in A$ the set $\calS^{\text{(i)}}_a$ contains at most one full rectangle $S$ satisfying
$\vertB(S)=b$: the one with minimum height. Recall that $\calS^{\text{(i)}}_a$ corresponds to the Down-Right scenario that assumes that the
rectangle immediately next to $S$ in the maximum weight independent set (which w.l.o.g.~consists only of full rectangles) is located to the right of
$S$. For this reason, the recursion in \eqref{recursion} still works if we redefine $\calS^{\text{(i)}}_a$ as the singleton containing the full
rectangle $S$ below $a$ minimizing $\vertB(S)_x$ ($\calS^{\text{(i)}}_a$ is empty if there is none). Such rectangle can be easily identified using the
matrices $M$ and $M'$: first, using $M$, we find the last element $b$ with $M(a,b)=1$; we define $\bar{b}$ as the element immediately after $b$ in
$B$ and finally, we look for the first row $\bar{a}$ such that $M'(\bar{a},\bar{b})=1$. The rectangle $S$ we were looking for is
$\Gamma(\bar{a},\bar{b})$. 
Similarly, for each $a'\in A$, $\calS^{\text{(ii)}}_a$ contains at most one full rectangle $S$ satisfying $\vertA(S)=a'$: the one with minimum height.
But since $\calS^{\text{(ii)}}_a$ corresponds to the Down-Down scenario, the recursion in \eqref{recursion} still works if we redefine
$\calS^{\text{(ii)}}_a$ as the singleton containing  the full rectangle $S$ below $a$ maximizing $\vertA(S)_y$. We compute all singletons
$\calS^{\text{(ii)}}_a$ in $O(n)$ time as follows. For each $b\in B$, we determine $T(b)$ using $M'$ (the minimum height full rectangle $S$ with
$\vertB(S)=b$ is given by the first 1-entry in the column of $M$ associated to $b$); then, we traverse the rectangles in $\mathcal{T}=\{T(b)\colon
b\in B\}$ increasingly with respect to $b_y$, keeping track of $T'(b)$ which we define as the rectangle $S\in \mathcal{T}$ weakly below $b$ with
highest $\vertA(S)_y$; finally, each singleton $\calS^{\text{(ii)}}_a$ can be computed in $O(1)$ by first finding the last element $b$ with
$M(a,b)=1$ and then returning $T'(\bar{b})$, where $\bar{b}$ is the element of $B$ immediately after $b$.

With $\calS^{\text{(iii)}}_b$ and $\calS^{\text{(iv)}}_b$ computed analogously, the $O(n)$ time algorithm is completed by solving the recursion
\eqref{recursion} on $V_{\downarrow }(a)$ and $
V_{\rightarrow}(b)$.
\end{proof}

Note that the assumptions made about the input in this theorem are not completely unrealistic. They hold, for example, when the ones of each row and column of $M$ and $M'$ are connected by a double linked list.

\subsection{Biconvex graphs}\label{sec:biconvex}
So far, we have identified a BRF graph $G$ with an arbitrary geometric representation $G=(A\cup B, \calR)$. The following result, valid for biconvex graphs, addresses  one  particular representation (note that in this representation different vertices may be mapped to points sharing coordinates).

\begin{theorem}\label{teo:biconvex}
Let $G=(A'\cup B', \calR')$ be a biconvex graph. Suppose that $A'=\{a'_1, \ldots, a'_s\}$ and $B'=\{b'_1, \ldots, b'_t\}$ are labellings of $A'$ and $B'$ so that the neighborhood of each $b'\in B'$ is a set of consecutive elements of $A'$ and vice-versa.
Map each $a'_i\in A'$ to the point $a_i=(i,-i)$ and each $b'_i \in B'$ to the point $b_i=(r(i),-l(i))$, where $l(i)$ (resp. $r(i)$) are the minimum (resp. maximum) index $j$ such that $\Gamma(a'_j,b'_i)\in \calR'$. Then $G=(A\cup B, \calR)$ is a representation of $G$ for which $\calI ( \calRD )$ is perfect.
\end{theorem}
\begin{proof}
We use the strong perfect graph theorem~\cite{chudnovsky2006strong}, proving by contradiction that $\calI ( \calRD )$ has no odd-holes nor odd-antiholes. 
First, suppose there is an odd-hole $\mathcal{H}=\{R_1,R_2,\ldots, R_k\}\subseteq \calRD$ formed by rectangles $R_l=\Gamma(a_{i_l},b_{i_l})$ that (only) intersect $R_{l-1}$  and $R_{l+1}$ (mod $k$). Assume that $a_{i_1}$ is the leftmost  defining corner. The three values $i_1,i_2$ and $i_k$ must be different: $i_2$ and $i_k$ are different because the corresponding rectangles do not intersect, while $i_1$ is different from $i_2$ (or $i_k$) since otherwise any rectangle intersecting the thinnest of these two rectangles would intersect the other one.  The rest of the argument, sketched below, refers to the lines $\lambda_1$ and $\lambda_2$ and the zones $Z_1,Z_2$, $Z_3$ and $Z_4$ defined in Figure~\ref{jump-biconvex} (left), where $i_1<i_2<i_k$ is assumed  without lost of generality. Zones $Z_1$ and $Z_4$ are closed regions, while $Z_2$ and $Z_3$ are open. The following claims can easily be verified:
\begin{itemize}
\item The point $\lambda_1\cap \lambda_2$ is in both $R_1$ and $R_k$. 
\item For all $1< l < k$, $i_l< i_k$, Indeed, the region $\cup_{j=2}^l R_j$ is connected, so it contains a continuous path from $a_{i_2}$ to $a_{i_l}$. Having $i_l> i_k$ would imply that such path crosses $\lambda_1$ or $\lambda_2$, and therefore some rectangle in $\{R_j\}_{j=2..l}$ intersects either $R_1$ or $R_k$ (contradicting that $\mathcal{H}$ is a hole). The equality $i_l=i_k$ could only hold if $l=k-1$, but $i_{k-1}< i_{k}$ by the same reason that $i_1<i_2$.
\item $i_2<i_{k-1}$, otherwise $R_2$ and $R_{k-1}$ would intersect.
\item Corners $b_{i_1}$ and $b_{i_k}$ lie in $Z_4$.
\item The corner $b_{i_2}$ lies in $Z_2$: being in any other zone would contradict the intersecting structure among the rectangles in the hole; being in $\lambda_2$ would contradict the inclusion-wise minimality of $R_1$. Analogously, $b_{i_{k-1}}$ lies in $Z_3$.
\item Corners $b_{i_3}, \ldots b_{i_{k-2}}$ lie in $Z_1$, otherwise the corresponding rectangles would intersect $R_1$ or $R_k$.
\item Either $b_{i_3}$ lies above $b_{i_{k-1}}$ or $b_{i_{k-2}}$ lies to the right $b_{i_{k-1}}$: if not, $R_2$ and $R_{k-1}$ would intersect. In what follows, set $j=i_3$ in the first case and $j=i_{k-2}$ in the second one.
\end{itemize}
Finally, observe that the four indices $j_1=i_2,j_2=j,j_3=i_{k-1}$ and $j_4=i_1$ are such that the associated first three intervals satisfy $l(j_1)<l(j_2)<l(j_3)$ and $r(j_1)<r(j_2)<r(j_3)$ while $l(j_4)<l(j_2)$ and $r(j_2)<r(j_4)$. It can be checked that no biconvex labeling of $B$ can comply with these inequalities, which gives the contradiction.

 \begin{figure}[ht]
\centering
\includegraphics[scale=0.8]{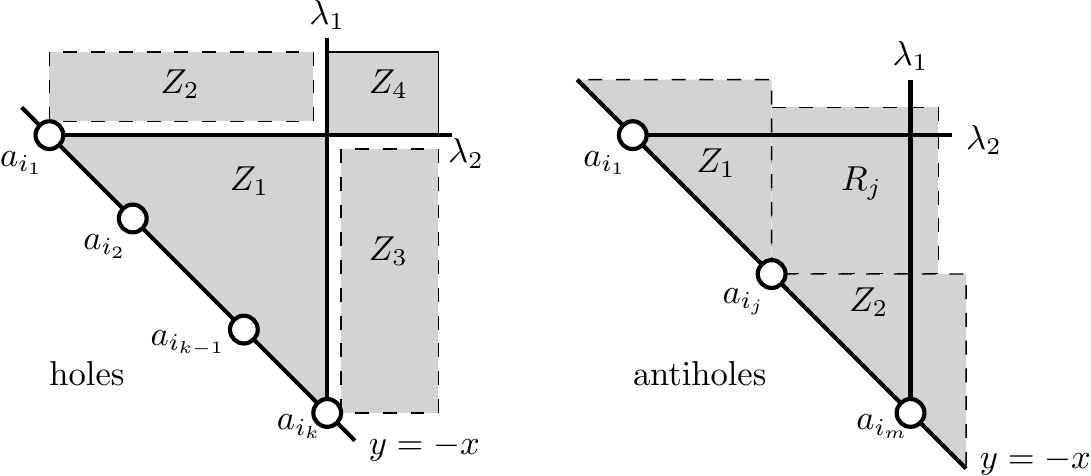}
\caption{Proof of Theorem \ref{teo:biconvex}.}\label{jump-biconvex}
\end{figure}

Now suppose the intersection graph $\calI(\calRD)$ has an odd-antihole $\mathcal{A}=\{R_1,R_2,\ldots, R_k\}$ of length at least 7. We keep the notation consistent from the odd-hole case; in particular, each $R_i$ does not intersect $R_{i-1}$ and $R_{i+1}$ (mod k). Assume that $a_{i_1}$ and $a_{i_m}$ are the leftmost and rightmost defining-corners in $A$, respectively. Using  Figure~\ref{jump-biconvex} (right) as reference, it is easy to see that.

\begin{itemize}
\item Rectangles $R_1$ and $R_m$ are the only ones with corners $a_1$ and $a_m$, respectively: same argument as with odd-holes.
\item Index $m\neq k$: if not, $R_4$ and $R_5$ would intersect both $R_1$ and $R_m$. Therefore, $\lambda_1 \cap \lambda_2$ lies in $R_4 \cap R_5$, which contradicts the intersecting structure of the antihole.
\item Rectangles $R_1$ and $R_m$ intersect: this follows from $m\neq k$ (above) and $m\neq 2$ which is proved in the same way.
\item There is a rectangle $R_j$ intersecting both $R_1$ and $R_m$ (because $k\geq 7$).
\end{itemize}
Rectangles $R_{j-1}$ and $R_{j+1}$ intersect each other, but do not intersect $R_j$, so they both lie either in zones $Z_1$ or in $Z_2$. This is a contradiction with the fact that $R_{j-1}$ and $R_{j+1}$ intersect $R_1$ and $R_m$.
 We conclude that $\calI(\calRD)$ has no odd-hole nor odd-antihole, and hence, it must be a perfect graph.
\end{proof}
In general, the intersection graph of inclusion-wise minimal rectangles in biconvex graphs is not always perfect (see Fig.~\ref{biconvexnonperfect}). But the previous construction shows that the \WMIS of biconvex graphs can be computed in polynomial time by solving a linear program.

\begin{figure}[ht]
\centering
\includegraphics[scale=0.8]{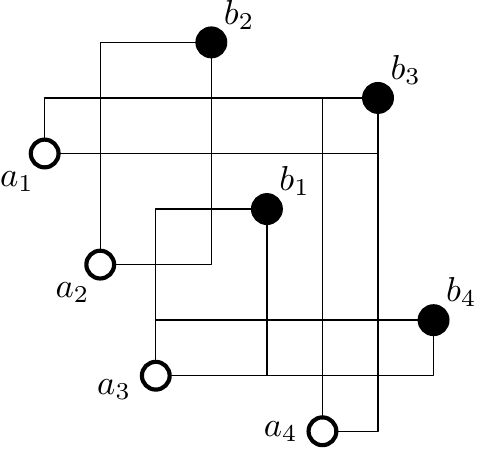}
\caption{The intersection graph $\calI(\calRD)$ of a biconvex graph is not always perfect.}\label{biconvexnonperfect}
\end{figure}
\subsection{Convex \dorgs}\label{subsec:convex}
Recall that convex bipartite \dorgs are the unrestricted \dorgs with $A\subset L$ and $B\subseteq L^+$ with $L$ the diagonal line $y=-x$.
Alternatively, they are the \dorg graphs $G=(A\cup B,\calR)$ such that the points in $A$ are incomparable under $\leq_{\RR^2}$.
As we discuss in the next section, the maximum weight independent set of convex \dorgs is equivalent to find the \defi{maximum weight point-interval set} of a collection of intervals. For the latter problem, described in Section \ref{sec:discussion}, Lubiw~\cite{Lubiw1991} provides a polynomial time algorithm  that directly translates into an $O(n^3)$-algorithm~\cite{soto2011jump} for \WMIS. Very recently, Correa et al.~\cite{CorreaFS14} improved Lubiw's result, obtaining a $O(n^2)$ algorithm.

\begin{theorem}\label{teo:lubiw}[Based on \cite{Lubiw1991,CorreaFS14,CorreaFPS14}] We can compute a \WMIS of a convex graph in $O(n^3)$ using Lubiw's algorithm for maximum weight point-interval set, and in quadratic time using the algorithm of Correa et al.~\cite{CorreaFS14}. 
\end{theorem}

\subsection{Interval bigraphs}

The natural geometric representation of an interval bigraph described at the beginning of the section is such that all the rectangles intersect the diagonal line $L$. We use this property and a recent result of Correa et al.~\cite{CorreaFPS14} to strengthen Theorem~\ref{thm:NPhard} as follows.

\begin{theorem} Computing a \WMIS of an interval bigraph is $\NP$-hard even for weights in $\{0,1\}$.
\end{theorem}
\begin{proof}
The problem of computing a \MIS of a family $\calI$ of rectangles intersecting $L$ is $\NP$-hard~\cite{CorreaFPS14}. Our hardness proof reduces from this problem, by transforming a collection of rectangles intersecting $L$ into a subset of rectangles of an interval bigraph (by translating and piecewise scaling the plane), and then using weights $\{0,1\}$ to distinguish the rectangles in the collection when solving the \WMIS. The proof is very similar to that of Theorem~\ref{thm:NPhard}.
\end{proof}

It is worth noting that currently, convex \dorgs is the largest natural class of \dorgs for which the \WMIS problem is solvable in polynomial time.  Nevertheless, Correa et al.~\cite{CorreaFS14, CorreaFPS14} gave a dynamic programming algorithm to compute \WMIS of families of rectangles intersecting a diagonal line having the following property: if two rectangles intersect then they share a point below the diagonal. Based on this, they devise a 2-approximation for \WMIS of rectangle families intersecting the diagonal whose running time is quadratic in the number of rectangles. Using their result we directly conclude that there exists an $O(n^4)$ time algorithm to compute a 2-approximation for the \WMIS of an interval bigraph.

\section{Discussion}\label{sec:discussion}
This section positions our results in the context of several closely related results for seemingly unrelated problems. In a nutshell, besides of greatly improving the algorithmic efficiency, our results greatly reduce the gap between the very complex algorithms, results and proofs for a generalization of our problem and the much simpler ones for a special case of our problem. 

Let $G=(A \cup B, \calR)$ be a biconvex graph (see Section \ref{sec:hierarchy}). Let $M^G$ be the biadjancency matrix whose rows and columns are sorted according to its corresponding biconvex labeling; note that the rows and columns of $M^G$ have their 1's in consecutive position. We can identify the entries $M^G_{a,b}$ of value 1 with the rectangles $\Gamma(a,b)$ in the geometric representation of $G$. We can also identify a hitting point $p$ with the set of entries $M^G_{a,b}$ corresponding to rectangles hit by $p$; it turns out that all those sets must be block matrices with entries of value 1. The following two equivalences are easy to prove in the biconvex case: 1) a collection of entries of value 1 in $M^G$ induces an independent set of rectangles in $G$ if and only if no block matrix contains two of such entries; and 2) a set of points $p$ defines a hitting set in $G$ if and only if its corresponding set of block matrices cover all the 1's in $M^G$. Chaiken et al.~\cite{ChaikenKSS81} show that the minimum size of a rectangle cover of a biconvex matrix $M$ (a set of block matrices covering all the 1's in $M$) equals the maximum size of an antirectangle (a set of 1's in $M$ such that no block matrix contains two of them); this corresponds to Theorem~\ref{teorem:teo1} for biconvex graphs. 


Going up to convex graphs, the work of Gy{\"o}ri~\cite{Gyori84} on point-interval pairs is particularly relevant. He works with a fixed ground set $A=\{1,\ldots,n \}$, intervals $I\subseteq A$ and point-interval pairs $(p,I)$ where $p\in I \subseteq A$. He introduces two notions: 1) a family of intervals $\mathcal{B}$ is a basis for another family of intervals $\mathcal{F}$, if every interval of  $\mathcal{F}$ can be written as union of intervals in $\mathcal{B}$; and 2) a collection of point-interval pairs $(p_i,I_i)_{i=1..m}$ is called independent if, for all $j\neq k$, either $p_j\notin I_k$ or  $p_k\notin I_j$. Gy{\"o}ri proves, non-constructively, that the cardinality of minimum basis for a family $\mathcal{F}$ equals the maximum cardinality of an independent family of point-interval pairs $(p,I)$, where $I$ is restricted to be in $\mathcal{F}$. To put this min-max result in our context, note that the containment relation on $(A \times \mathcal{F})$ defines a convex graph on $A$. Through the representation of convex \dorgs used in Section~\ref{subsec:convex}, we can represent $A$ as a set of points in the antidiagonal line $y=-x$, while we can represent  $\mathcal{F}$  as points weakly above this line. The rectangles in the convex graph $(A\cup \mathcal{F},\calR)$ become the set of all point-interval pairs $(p,I)$ where $I\in F$.
It is easy to see that independent set of rectangles are in  correspondence with independent families of point-interval pairs.  
Hitting sets of rectangles are also in correspondence with minimum basis of intervals, although not bijectively: we identify $q$ with the interval $I_q$  of points $a\in A$ with $a \leq_{\RR^2} q$; a hitting set $H$ is transformed into a basis $\{I_q: q\in H\}$ for $\mathcal{F}$ since $I=\cup_{q\in H, q\leq_{\RR^2} I} I_q$ for all $I\in F$. On the other hand, a basis for $\mathcal{F}$ induces a hitting set $H$ for $\calR$ of the same cardinality, through the inverse identification.

The equivalences just described show that Gy{\"o}ri's min-max result on point-interval pairs implies Theorem~\ref{teorem:teo1} for convex graphs, and also that Lubiw's algorithm~\cite{Lubiw1991} for weighted independent set of point-interval pairs implicitly provides an algorithm for the maximum weighted independent set of rectangles in convex \dorgs, as described in Theorem~\ref{teo:lubiw}. Furthermore, Gy{\"o}ri~\cite{Gyori84} also shows that for convex matrices (i.e., where the 1's on each column are consecutive), the minimum size of a rectangle cover equals the maximum size of an antirectangle, thus extending the min-max result of Chaiken~\cite{ChaikenKSS81}. This indeed follows from Theorem~\ref{teorem:teo1} for convex graphs, through a simple geometric argument we skip here~\cite{Gyori84}. 

Following Gy{\"o}ri's non-constructive proof, there was significant interest in obtaining a constructive, simple and efficient version of his min-max result, and their generalizations. Franzblau and Kleitman~\cite{FranzblauK84STOC} present an algorithmic proof that uses the original ideas from Gy{\"o}ri. Later, Frank~\cite{Frank99-2} presents an alternative algorithmic proof using new ideas, which form the core of Algorithm~\ref{alg:mis-mhs}: interpreting his algorithm in our geometric setting, the procedure starts from a convex \dorg $G=(A\cup B, \calR)$, determines a cross-free intersection family $\calK \subseteq \calR$, and then uses Dilworth's theorem to determine a maximum independent set and a hitting set of $\calK$, which then manages to transform into a maximum independent set and a hitting set of $\calR$. 
After the additional improvements in~\cite{benczur1999dilworth}, the theoretical performance\footnote{In~\cite{benczur1999dilworth}, running times are expressed in terms of $\abs{A}$ and $\abs{B}$, whereas we measure in terms of $n=\abs{A \cup B}$.} of the algorithm of Franzblau and Kleitman is better than the one of Frank,  providing running times of $O(n^2)$ and $O(n^{2.5} \log^{3/2} n)$ for the minimum hitting set and the maximum independent set of convex \dorgs,  respectively.

A concrete contribution of our work is to adapt the Frank's algorithm to the much broader set of \dorgs, with  no additional overhead in the running time. We remark that although we use some ideas from~\cite{benczur1999dilworth} to tweak the algorithm, our case is significantly more complex.

The algorithm of Frank was not developed only in pursuit of a clean algorithmic version of the theorem of Gyori (which anyway was already given by Franzblau and Kleitman). A few years earlier, Frank and Jord{\'a}n~\cite{FrankJ95} had presented a much more general min-max theorem which encompassed all the special cases above, including the case of \dorgs. The description of this theorem, in principle, is much more abstract.  A collection of pairs of sets $(S_i, T_i)$ is \defi{half-disjoint} if for every $j \neq k$, either $S_j \cap S_k$ or $T_j \cap T_k$ are empty.  A directed-edge $(s, t)$ \defi{covers} a set-pair $(S, T)$ if $s \in S$ and $t \in T$. A family $\mathcal{S}$ of set-pairs is \defi{crossing} if whenever $(S, T)$ and $(S', T')$ are in $\mathcal{S}$, so are $(S \cap T, S' \cup T')$ and $(S \cup T, S' \cap T')$. Frank and Jord{\'a}n  show that for every
crossing family $\mathcal{S}$, the maximum size of a half-disjoint subfamily is equal to the minimum size of a collection of directed-edges covering $\mathcal{S}$. Even further, the proof they offer was algorithmic, even though it relied on the ellipsoid method. 

Our min-max result follows from Frank and Jord{\'a}n's theorem: the inclusionwise minimal rectangles $\calRD$ of a \dorg $(A\cup B, \calR)$, once projected over both axes $\{(R_x,R_y): R\in \calRD\}$, becomes a crossing family of set-pairs for which half-disjoint subfamilies become independent sets in $\calR$, while coverings by directed edges become hitting sets in $\calR$. Thus, it may seem that our contribution is merely that of an algorithmic improvement. Yet, the literature that followed~\cite{FrankJ95} shows why this is not the case. The generality of Frank and Jord{\'a}n also carries significantly more abstract concepts, algorithms and proofs. Already the work of Frank ~\cite{Frank99-2} and Bencz\'{u}r et al.~\cite{benczur1999dilworth} show that a significant effort is required in order to translate the original ideas of Frank and Jord{\'a}n~\cite{FrankJ95} into an efficient and intuitive algorithmic proof for the theorem of Gy{\"o}ri. More recently, the combinatorial algorithm of Benczur~\cite{Benczur03} for pairs of sets gives a more intuitive view of these objects, but still both the algorithm and the  analysis are still much more complex and abstract than ours. The algorithmic proof we provide for Theorem~\ref{teorem:teo1} has the value of positioning \dorgs at the same level of complexity  than the convex \dorgs studied by Gy{\"o}ri, both conceptually and algorithmically, for the problems we are concerned here.

It is worth noting that the  min-max result also apply to \dorgs that are drawn in a cylinder $\mathbb{S}^1\times \RR$ or a torus $\mathbb{S}^1 \times \mathbb{S}^1$. In both surfaces axis-aligned rectangles are well-defined as cartesian products of closed intervals. Given two finite sets of points $A$ and $B$ and an arbitrary set $\calZ$ in a surface $\mathcal{S}$ that can be either a cylinder or a torus, we can still define the collection $\calRD$ of inclusionwise minimal axis-aligned rectangles contained in $\calZ$ with lower left corner in $A$ and upper right corner in $B$. It is easy to see that $\{(R_x,R_y) \colon R_x\times R_y \in \calRD\}$ is a crossing family of set-pairs. Applying Frank and Jord\'an theorem, the size of a maximum independent  set in $\calRD$ equals the size of a minimum hitting set. We believe it is not hard to modify our combinatorial algorithm to work in this case too, but we defer this to future work. 

\section{Acknowledgements}

The first author was partially supported by Nucleo Milenio Informaci\'on y Coordinaci\'on en redes ICM/FIC P10-024F  and by CONICYT via FONDECYT grant 11130266. The second author was partially supported by the FSR Incoming
Post-doctoral Fellowship of the Catholic University of Louvain (UCL), funded by the French Community of Belgium. 
The authors gratefully acknowledge Prof.~Andreas S.~Schulz for many stimulating discussions during the early stages of this paper.





\bibliographystyle{elsart-num-sort}
\bibliography{rectangles-journal}







\end{document}